\documentclass[format=acmsmall, review=false]{acmart}
\usepackage{acm-ec-24}
\usepackage{booktabs} 
\usepackage[ruled]{algorithm2e} 
\usepackage{threeparttable}

\usepackage{tikz}
\usetikzlibrary{shapes}
\usetikzlibrary{calc}
\usetikzlibrary{decorations.pathreplacing}
\SetAlFnt{\small}
\SetAlCapFnt{\small}
\SetAlCapNameFnt{\small}
\SetAlCapHSkip{0pt}
\IncMargin{-\parindent}

\setcitestyle{authoryear}

\title[Menu-Based Deep Auctions]{\sc{GemNet}: Menu-Based, Strategy-Proof Multi-Bidder Auctions Through Deep Learning}

\author{Tonghan Wang$^\ast$}
\email{twang1@g.harvard.edu}
\author{Yanchen Jiang$^\ast$}\email{yanchen\_jiang@g.harvard.edu}
\author{David C. Parkes}\email{parkes@eecs.harvard.edu}

\affiliation{
  \institution{Harvard University}
  \city{Cambridge}
  \state{MA}
  \country{USA}
}

\begin{abstract}
Automated mechanism design (AMD) uses computational methods for mechanism design. 
\emph{Differentiable economics} is a form of AMD that uses deep learning to learn mechanism designs and has enabled strong progress in AMD in recent years. Nevertheless, a major open problem has been to learn multi-bidder, general, and fully strategy-proof (SP) auctions. We introduce \emph{GEneral Menu-based NETwork} (\name), which significantly extends the menu-based approach of  the single-bidder RochetNet~\citep{duetting2023optimal} to the multi-bidder setting. The challenge in achieving SP is to learn bidder-independent menus that are feasible, so that the optimal menu choices for each bidder do not over-allocate items when taken together (we call this \emph{menu compatibility}). \name\ penalizes the failure of menu compatibility during training, and transforms learned  menus after training through price changes, by considering a set of discretized bidder values  and reasoning about Lipschitz smoothness to guarantee menu compatibility on the entire value space. This approach is general, leaving trained menus that already satisfy menu compatibility undisturbed and reducing to RochetNet for a single bidder. Mixed-integer linear programs are used for menu transforms, and through a number of optimizations enabled by deep learning, including adaptive grids and methods to skip menu elements, we scale to large auction design problems. \name~learns auctions with  better revenue than affine maximization methods, achieves exact SP whereas previous general multi-bidder methods are  approximately SP, and offers greatly enhanced interpretability.
\end{abstract}

\usepackage{amsmath}
\usepackage{mathtools}
\usepackage{amsthm}
\usepackage{wrapfig}

\usepackage[capitalize,noabbrev]{cleveref}

\makeatletter
\def\thmheadbrackets#1#2#3{%
  \thmname{#1}\thmnumber{\@ifnotempty{#1}{ }\@upn{#2}}%
  \thmnote{ {\the\thm@notefont[#3]}}}
\makeatother

\newtheoremstyle{brakets}
  {}
  {}
  {\itshape}
  {}
  {\bfseries}
  {.}
  { }
  {\thmheadbrackets{#1}{#2}{#3}}

\theoremstyle{brakets}

\newtheorem{theorem}{Theorem}

\newtheorem{lemma}[theorem]{Lemma}

\newtheorem{definition}[theorem]{Definition}

\theoremstyle{remark}

\definecolor{darkgreen}{rgb}{0.0, 0.5, 0.0}
\definecolor{darkblue}{rgb}{0.0, 0.5, 1.0}
\usepackage[textsize=tiny]{todonotes}

\newcount\Comments  
\newcommand{\kibitz}[2]{\ifnum\Comments=1{\color{#1}{#2}}\fi}

\newcount\CommentsAdd  
\newcommand{\kibitzAdd}[2]{\ifnum\CommentsAdd=1{\color{#1}{#2}}\fi}
\definecolor{english}{rgb}{0.0, 0.5, 0.0}
\definecolor{tw}{rgb}{0.0, 0.0, 0.5}

\newcommand{\tw}[1]{\kibitzAdd{blue}{#1}}
\newcommand{\jfadd}[1]{\kibitzAdd{blue}{#1}}

\usepackage{xcolor}

\usepackage{amsmath,amsfonts,bm}

\def\eqref#1{equation~\ref{#1}}

\def\1{\bm{1}}

\def\vb{{\bm{b}}}

\def\vv{{\bm{v}}}

\def\vz{{\bm{z}}}

\DeclareMathAlphabet{\mathsfit}{\encodingdefault}{\sfdefault}{m}{sl}
\SetMathAlphabet{\mathsfit}{bold}{\encodingdefault}{\sfdefault}{bx}{n}

\DeclareMathOperator*{\argmax}{arg\,max}

\newcommand{\shortn}{\textup{\texttt{-}}}

\newcommand{\shortp}{\textup{\texttt{+}}}

\newcommand{\Tau}{\mathrm{T}}
\newcommand\spname[1]{$\mathtt{#1}$}
\newcommand{\name}{\textsc{GemNet}}
\newcommand{\namekp}{KC}

\newcommand{\neighborjoint}{$\mathcal{N}_{\frac{\epsilon}{2}}$}

\usepackage{multirow}
\usepackage{makecell}
\newcolumntype{L}{>{$}l<{$}}
\newcolumntype{C}{>{$}c<{$}}
\newcolumntype{R}{>{$}r<{$}}

\begin{document}

\begin{titlepage}
\renewcommand{\shortauthors}{Tonghan Wang, Yanchen Jiang, and David C. Parkes}
\maketitle\makeatletter \gdef\@ACM@checkaffil{}
\makeatother
\setcounter{tocdepth}{2} 
\renewcommand{\thefootnote}{\relax}
\footnotetext{\normalsize$^\ast$\footnotesize Equal Contribution.}
\renewcommand{\thefootnote}{\relax} 
\footnotetext{\vspace{0pt}\noindent\hrulefill\vspace{0pt} \\[0pt] This paper received the \textit{Exemplary Paper Award} for the AI track at the Twenty-Fifth ACM Conference on Economics and
Computation (ACM EC ’24), where it appeared as an extended abstract.}
\renewcommand{\thefootnote}{\arabic{footnote}}

\end{titlepage}

\section{Introduction}

Auctions stand as one of the most enduring and thriving economic activities to this day~\citep{cramton2006combinatorial,mcafee1996analyzing,milionis2023myersonian,edelman2007internet,milgrom2004putting}, and the design of revenue optimal auctions forms a cornerstone problem in economic theory.
%
 The seminal work~\citep{myerson1981optimal} solves optimal auctions for selling a single item.
 While the long-standing significance of auctions is well-established, it is notable that decades of theoretical exploration have yet to fully unravel the intricacy of optimal auction design.
Considering dominant-strategy incentive compatibility (DSIC, also strategy-proof) auctions, the only analytical results are for variations on the single-bidder setting~\citep{daskalakis2015strong,manelli2006bundling,pavlov2011optimal, giannakopoulos2014duality} and for multiple bidders for the case of two items and value distributions with only two possible values in
their support~\citep{yao2017dominant}.
%


The idea of using computational methods to design auctions and mechanisms, coined {\em automated mechanism design} (AMD), was introduced by~\citet{conitzer2002,conitzer2004}.
Whereas the early literature on AMD~\citep{conitzer2002,conitzer2004} proposed using linear programming (LP) for mechanism design in discrete value domains, a general LP formulation for strategy-proof mechanisms in continuous value domains has remained elusive \citep{conitzer2006computational}. Moreover, even in discrete settings, these LP-based formulations suffer from exponential size, making them computationally impractical as the size of the type space increases.
In recent years, \emph{differentiable economics}~\citep{duetting2023optimal} 
has proposed using deep learning to discover  optimal, strategy-proof auctions and optimal mechanisms more generally. Machine learning pipelines based on neural networks are used as computational techniques for optimizing within function classes, where the functions represent auction or mechanism rules. This methodology aligns with the paradigm of unsupervised learning by employing appropriately defined loss functions and sampling training data from a known type distribution. Ideally, this approach when used to auction design
would satisfy the following three properties: (1) {\em expressive}: revenue-optimal auctions are in the function class represented by the neural network; (2) {\em strategy-proof}: the learned auctions are  DSIC; and (3) {\em multi-bidder}: the framework is able to support multi-item auctions with multiple bidders.

\begin{wrapfigure}[15]{r}{0.52\linewidth}
\centering
  \includegraphics[width=\linewidth]{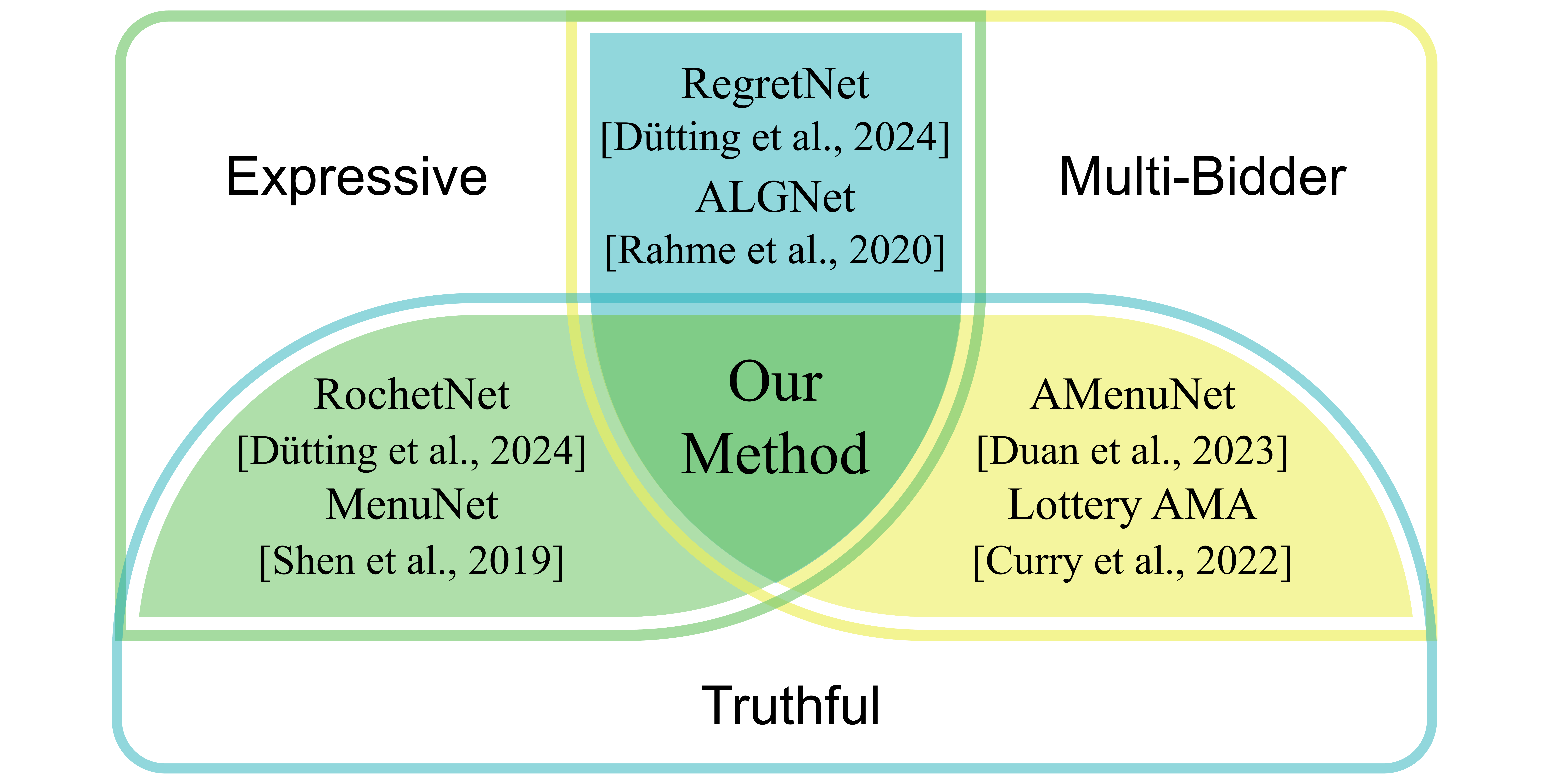}
  \captionsetup{font=footnotesize}
    \caption{{\sc GemNet} is the first differentiable economics method
 that is   generally expressive, truthful (or strategy-proof, DSIC),
 and supports multiple bidders (and items).
    }\label{fig:position}
\end{wrapfigure}

Finding a method that meets all  three criteria is a complex challenge, as illustrated in Fig.~\ref{fig:position}. RegretNet~\citep{duetting2023optimal}, along with subsequent developments such as ALGNet~\citep{rahme2020auction} and RegretFormer~\citep{ivanov2022optimal}, 
are expressive and multi-bidder but seek to minimize the DSIC violation during training 
and lack a guarantee of exact strategy-proofness (SP).
Affine-maximizer auctions (AMAs)~\citep{curry2022differentiable,duan2023scalable} ensure SP
 by appealing to generalized, affine-maximizing Vickrey-Clarke-Groves mechanisms~\citep{vickrey1961counterspeculation, clarke1971multipart, groves}.
 Although applicable to multi-bidder scenarios, AMAs are not fully expressive because they are restricted to affine-maximizing mechanisms.
Menu-based methods such as RochetNet~\citep{duetting2023optimal} and MenuNet~\citep{shen2019automated} are expressive and SP, 
 but only support auctions with a
 single bidder (or settings where the items are non-rival, for example digital content,  where an auction
 with multiple bidders 
 decomposes into  single-bidder problems).

\begin{figure}
    \centering
    \includegraphics[width=\linewidth]{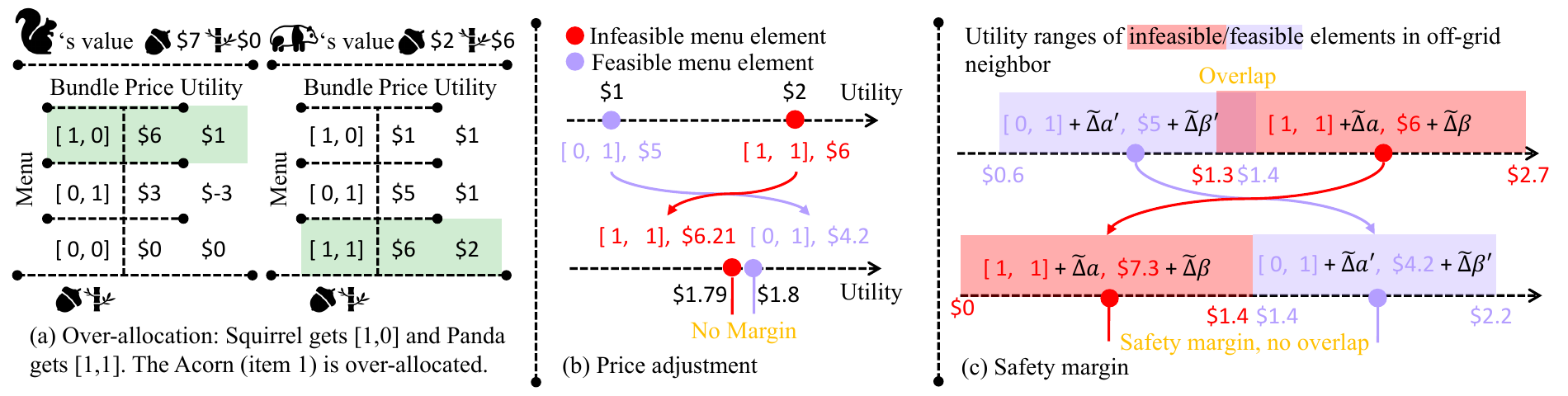}
    \caption{Illustrating our method for additive valuations. (a) An example of incompatible menus where the utility-maximizing menu elements over-allocate the first item. (b) Prices are adjusted so that the feasible menu element, $[0,1]$ in this case, has greater utility than infeasible elements like $[1,1]$. Menus after price adjustment are compatible at grid points.  (c) Extending menu compatibility to off-grid values by introducing safety margins. In this example, the utility of an infeasible bundle $[1,1]$ changes from \$1.3 to \$2.7 between two grid points. Here, the safety margin represents the gap between the utilities of feasible and infeasible menu elements at grid points preventing their utility ranges from overlapping.}
    \label{fig:demo}
\end{figure}
%
An idea for achieving all three criteria is to generalize the RochetNet architecture to multiple bidders while preserving its core features of {\em self-bid independent menus} (each bidder faces choices that may depend on the bids of others
but not its own bid) and {\em bidder-optimization} 
(each bidder receives an optimal menu choice given its own bid).
This kind of menu-based structure is inherent in SP mechanisms: 
any SP
 mechanism can be represented by self-bid independent menus with bidder-optimization~\citep{hammond1979straightforward}.
%
%
The  challenge in the multi-bidder setting is to ensure that 
 bidder-optimizing menu choices, one for each agent,  do not
 over-allocate items  when taken together. We refer to this
 feasibility property as {\em menu compatibility}. Fig.~\ref{fig:demo}(a) shows an example where over-allocation occurs. Two bidders, \spname{Squirrel} and \spname{Panda}, are each provided with a menu. Each bidder selects the menu item that maximizes their utility: $[1,0]$ for \spname{Squirrel} and $[1,1]$ for \spname{Panda}. However, these choices are incompatible because they over-allocate the first item \spname{Acorn}. 


 
 %
%
Bridging this gap to enable a multi-bidder RochetNet is important for a number of  reasons. First, it  may reveal a new 
structural understanding of optimal auction design, both through the 
interpretability of menu-based descriptions  and 
from achieving the distinguished  property of 
{\em exact} strategy-proofness.
Second, it  drives progress in  understanding  
how well existing solutions such as those coming from AMAs 
and approximately-SP methods such as RegretNet 
 align with  optimal  designs. Third, it may provide new opportunities for operational impact, for example in 
the kinds of systems (advertising and marketing technology, notably) that rely on auctions, through the additional trust that comes through exact SP.

\textbf{Our Contributions}.\ \ We close this gap and  propose the first differentiable economics
method that fulfills all  three criteria. Our method, {\em GEneral Menu-based NETwork} (\name),
makes a significant generalization of  RochetNet  to the multi-bidder, multi-item setting under additive and unit-demand valuations, and is able to guarantee menu compatibility. A neural network is used to learn a 
self-bid independent menu  for each bidder, consisting of a set of bundle-price pairs,
each pair representing a  menu choice 
and corresponding to a bundle of items---actually a distribution on items---and a price.
As we explain, the prices associated with the learned menus are  transformed after
training to ensure menu compatibility over the full domain and SP.  Given these adjusted
menus, the outcome of an auction
corresponds to a bidder-optimizing menu choice for each bidder (i.e.,
a choice from the bidder's menu that maximizes its utility 
based on its  value  for different bundles 
as represented through its bid).
%
For a network with sufficient capacity (i.e., enough hidden layers and nodes), this menu-based computational framework is without loss of generality, in that it can
learn the revenue-optimal auction with access to sufficient training data,
due to the
universality of menu-based representations~\citep{hammond1979straightforward}.
In the special case of a single bidder, \name\ reduces to RochetNet.

 %
The primary issue in achieving SP with a menu-based, multi-bidder architecture 
is  to guarantee menu compatibility. 
Straightforward approaches such as scaling down each bidder's randomized assignment of an item according to the over-assignment  of the item  fails strategy-proofness. This would 
create a dependency between  a first bidder's scaled
menu and the choices of a second bidder, where the choice of this second bidder depends on the 
first bidder’s own reports through the design of the menu for the second bidder.
This interdependence contravenes the conditions of self-bid independence and 
leads to violations of  SP.
%
\begin{figure*}
    \centering
    \includegraphics[width=\linewidth]{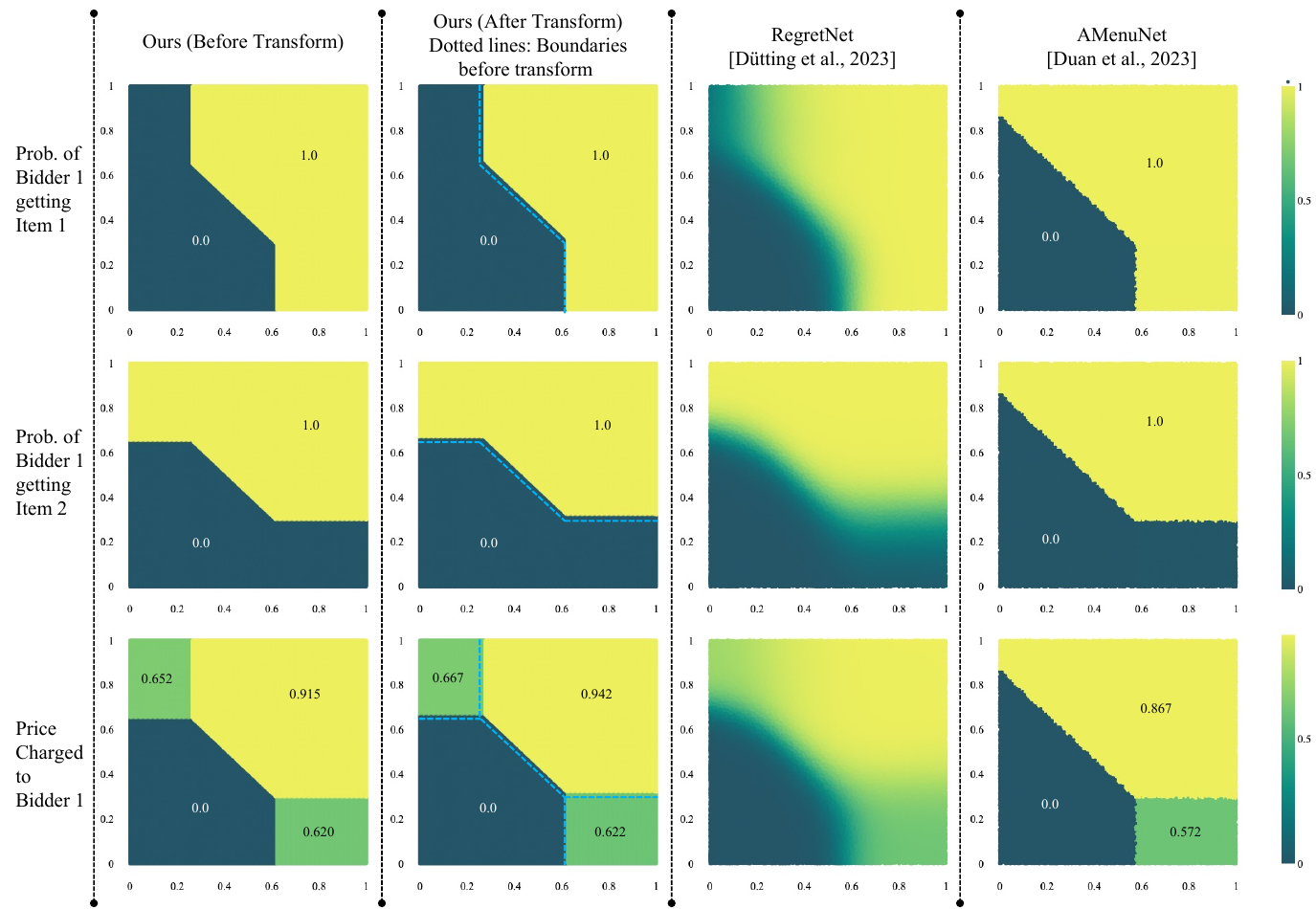}
    \caption{Learned mechanisms in the setting with two additive bidders, two items, and i.i.d.~uniform values on $[0,1]$. 
Bidder 2's values are set at $(0, 0.6)$, and the x- and y-axis in each subplot is Bidder 1's value for Item 1 and 2. Rows 1 and 2 show the probability of Bidder 1 getting Item 1 and Item 2, respectively,  and row 3
shows the price for Bidder 1. Columns represent different methods. Dotted lines in the 
second column distinguish the pre- and post-transformation mechanisms. Compared to RegretNet, \name~exhibits a clear decision boundary and thereby improves the interpretability of designed mechanisms. AMenuNet has a large set of types in the left-top region where Bidder 1 receives no items and makes no payments. In contrast, \name~increases revenue by allocating Item 1 to Bidder 1 within this region, highlighting the improved expressive capacity enabled by our method.
\label{fig:mech_compare_2x2} }
\end{figure*}

In achieving menu compatibility,  we first  develop a suitable loss function to use during training that represents  the goal of maximizing revenue while penalizing loss of menu compatibility. Specifically, we design an \emph{incompatibility loss} that penalizes menus 
 for which the simultaneous choice of bidder-optimizing menu  elements for each  bidder over-allocates items. 
 This incompatibility loss is used  together with a loss term corresponding to negated revenue during training. Although
  the incompatibility loss does not completely
 address menu incompatibility, it substantially lowers the likelihood of over-allocation. 
After training, we then modify as needed the prices associated with the elements of menus, with this menu transformation used to ensure the exact menu compatibility of a deployed  mechanism over the full value domain. Menu transformation makes small changes to prices so that incompatible choices become less appealing
to a bidder than the best compatible choice.
%
 %

 The post-training price transformation involves solving a series of
 mixed-integer linear programs (MILP).
Although the method works on a grid of values for bidders, through Lipschitz smoothness arguments it guarantees menu compatibility on the full, continuous input domain.
 For  bidder $i$, the transform considers a grid of values of the others. Fixing one such value profile $\tilde{v}_{-i}$,  and thus  fixing the menu to bidder $i$,
the transform considers a grid $\mathcal{V}_i$
 of possible values $v_i$ for bidder $i$. This set of values $\mathcal{V}_i$ induces a
 MILP corresponding to $(i,\tilde{v}_{-i})$, where the decision variables are
 the possible price changes   to make
 to certain elements in bidder $i$'s menu at $\tilde{v}_{-i}$. 
To construct this MILP,
 we consider the effect of assuming each $v_i\in \mathcal{V}_i$ as the valuation
 of bidder $i$. As $v_i$ changes, so too does the menu
 of choices offered to each of the other bidders, and for each $v_i$
 we identify the bidder-optimizing menu choice for each of the other bidders at
 $\tilde{v}_{-i}$.
These menu choices by $j\neq i$  may render some elements in bidder $i$'s menu infeasible for this choice of $v_i\in \mathcal{V}_i$, 
while other elements remain feasible.  
For each $v_i\in \mathcal{V}_i$, the MILP includes constraints on
price changes 
for  elements of the menu at $\tilde{v}_{-i}$ to ensure that  
bidder $i$ prefers some feasible  element over every 
infeasible element, with an
objective to minimize price adjustments. 


Similar MILPs are formulated and solved in series for different values $v_{-i}$ of other bidders, then the process proceeds by adopting the viewpoint of each bidder in turn, using the transformed menus for the bidders already considered and making adjustments to their learned menus as needed.
Each transform is always feasible because in the worst-case it can make all the menu elements so expensive that a bidder prefers to choose the empty bundle at zero price. 

Fig.~\ref{fig:demo}(b) shows an example of price adjustments considering a grid of values. We fix the bids of \spname{Squirrel} and consider a grid of values for \spname{Panda}. When \spname{Panda} bids $[2, 6]$, we keep \spname{Squirrel}'s menu unchanged so it continues to choose  bundle $[1,0]$. This renders only one item in \spname{Panda}'s menu feasible. We can change the prices in \spname{Panda}'s menu to make this feasible bundle $[0,1]$ the most attractive option. This process is repeated for all \spname{Panda}'s values, setting up a MILP to determine the specific price changes. We then repeat the procedure for \spname{Squirrel} values in a grid, solving separate MILPs for each case.

This illustration of the method in Fig.~\ref{fig:demo}(b) guarantees menu compatibility over the grid. To extend this compatibility to the entire continuous value space, we also need to accommodate utility changes at off-grid bids. For this, we employ network Lipschitz smoothness to bound the menu changes between grid points, leading to estimated utility ranges for off-grid values. The MILP ensures that the utility of feasible elements exceeds that of infeasible elements by a \emph{safety margin} at a grid point. This margin prevents utility ranges of these elements from overlapping, ensuring that infeasible elements cannot be selected at off-grid values. An example is shown in Fig.~\ref{fig:demo}(c). Compared to the scenario in Fig.~\ref{fig:demo}(b), the required price changes in this case are larger.

At deployment, and for a  bid profile $v$ that may not  be on the grid and some bidder $i$, we feed other bids $v_{-i}$ to the network and change the network-generated prices by adding the price adjustment from the closest grid point in terms of the $\ell_\infty$-norm. From this adjusted menu, bidder $i$ chooses the utility-maximizing bundle.
%

%
 
We prove that by working with safety margins, this process of menu adjustment guarantees menu compatibility across the entire, continuous value domain (Theorem~\ref{thm:2_bidder_comp_updated} and Theorem~\ref{thm:general_comp}) and SP (Theorem~\ref{thm:sp}). The proof also  supports an approach that reduces the grid size by adapting it, locally, to the local Lipschitz constant. Additionally, given that the trained networks exhibit only a small over-allocation rate, it becomes feasible to maintain bidders' choices in the original, trained menu (before adjustment) across the majority of grid points.
 These strategies, among others, enable a significant reduction in the number of binary variables and  the time needed to solve these MILPs.
 For example, we reduce the number of binary variables from 548,866 to 28 on average for an auction with 2 additive bidders, 2 items, and i.i.d.~uniform values from $[0,1]$, which saves $>99.99\%$ of the running time (Table~\ref{tab:eff_strategy}).

The price transformation   leaves learned menus
undisturbed when they are  already menu-compatible. Moreover, the
 transformation    is not manipulable through misreports
 because  the grid points are self-bid independent and we
 ensure menu compatibility across the  entire value domain
and thus bidders cannot benefit by intentionally triggering or avoiding
transformation on some part of the domain.
%

Given that early AMD methods used LP~\citep{conitzer2002,conitzer2004,conitzer2006computational} but were restricted in discrete value domains, it is intriguing to ask whether our price-adjustment method could extend to an MILP-only formulation for SP mechanisms in continuous value domains without relying on deep learning. This question is of broader interests in the AMD community. We find that, although it is possible to formulate such MILPs using our ideas like safety margins and menu compatibility, several challenges need to be addressed to make this formulation practical. For example, deep learning initialization allows us to solely focus on price changes; without it, additional decision variables for bundles must be introduced. Moreover, our techniques that significantly accelerate MILPs depend on deep learning initialization. In its absence, MILP-only methods may encounter prohibitively long running times. Until these challenges are resolved, deep learning remains a more scalable solution with manageable complexity.

%
%
%
\begin{figure}
    \centering
    \includegraphics[width=\linewidth]{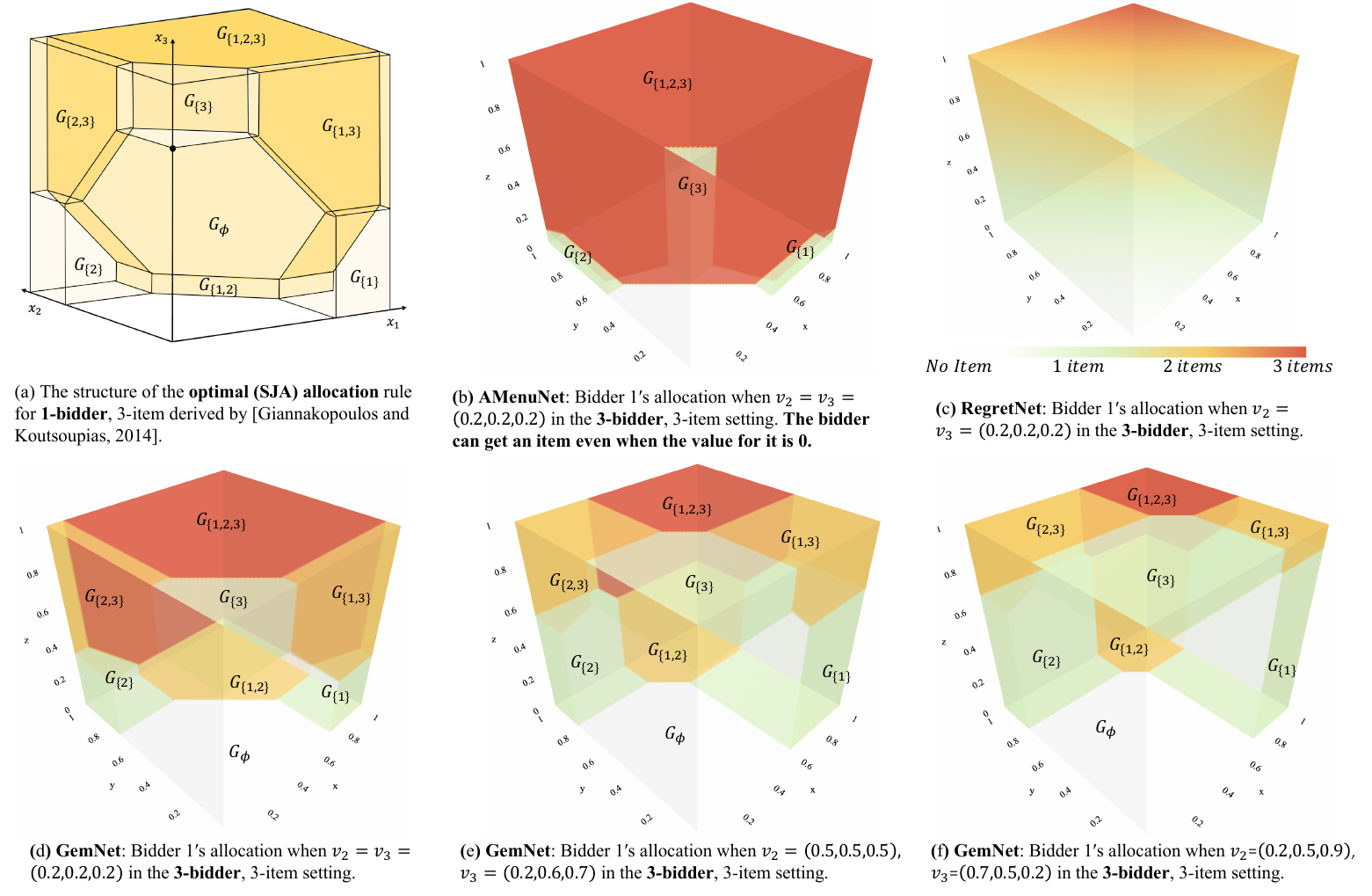}
    \vspace{-2em}
    \caption{Three items, additive bidders, i.i.d.~uniform values on $[0,1]$.
 The value of a particular bidder for each of three items varies in each cube,
    and annotation $G_{\{S\}}$ means the bidder gets  items in set $S$ in a region.
{\bf (a) One bidder.} The optimal allocation structure for the 1-bidder case~\citep{giannakopoulos2014duality}.
  An optimal analytical solution is not known for two or more bidders.
   {\bf (b-d) Three bidders.} Fix the values of Bidders 2 and 3 to $v_2=v_3=(0.2,0.2,0.2)$, showing the allocation for Bidder 1 learned by AMenuNet~\citep{duan2023scalable}, RegretNet~\citep{duetting2023optimal}, and~\name. AMenuNet learns a sub-optimal mechanism, e.g., Bidder 1 gets Item 2 even when its value is 0 (when $v_1(2)=0$, $v_1(1)>0.5$, $v_1(3)>0.2$). It is interesting that \name~obtains an allocation rule conforming to the optimal structure in the 1-bidder case. {\bf (e-f)  Three bidders.}
 \name~adapts the allocation for Bidder 1, while maintaining the 
     highlevel  structure, for different values of Bidders 2 and 3. \label{fig:u33_compare}}
     \vspace{-1.5em}
\end{figure}

The experimental results first provide a detailed investigation of the auctions learned by \name, these set against strong baseline algorithms. The comparison with RegretNet~\citep{duetting2023optimal} demonstrates that \name~brings the advantage of menu-based auction design to multi-bidder settings, with clear and interpretable decision boundaries. The comparison with affine-maximization methods~\citep{duan2023scalable}  shows how the improved expressive capacity of \name~leads to a different mechanism with higher revenue (see Fig~\ref{fig:mech_compare_2x2}, for example). Working in a 3-bidder, 3-item, and uniform valuation setting 
and comparing to the  Straight-Jacket Auction (SJA)~\citep{giannakopoulos2014duality} (Fig.~\ref{fig:u33_compare}), which is optimal for the single-bidder, 3-item, and uniform valuation setting,
 we see 
that \name~generates an auction design for multiple bidders
that is similar to 
 the analytically optimal single-bidder solution  when the other two bidders have
 identical valuations
 and reveals new suggested optimal solutions for asymmetric settings.
Compared with 
 RochetNet~\citep{duetting2023optimal} and AMenuNet~\citep{duan2023scalable},  it is plain that \name~provides enhanced interpretability on the design of optimal auctions.

We also benchmark \name\ on various auction settings studied in the recent differentiable economics literature, scaling up to 8 bidders and 10 items. \name\ consistently outperforms the affine-maximization methods
across all  settings, and gets revenue  slightly below that achieved by RegretNet, which does not generate exactly SP auctions.
Another validation is that \name\ is  the only deep learning method that can recover what,
 to our knowledge, is the only  existing analytical solution
for a multiple-bidder, multi-item SP auction~\citep{yao2017dominant}.
 These findings suggest that \name\ offers improved expressiveness 
 and substantially advances   the state-of-the-art performance that
 can be accomplished through differentiable economics and AMD more generally.




\subsection{Additional related work}

In the context of {\em Bayesian incentive compatibility} (BIC), 
 there have been  advances in automated mechanism design through the use of 
{\em interim} representations and characterization results~\citep{babaioff2020simple,cai2012algorithmic,Cai2012,daskalakis2012symmetries,hartline2011bayesian}. A key
 challenge has been
  to transform an $\epsilon-$BIC mechanism, constructed on a discrete (interim) grid
to an exact BIC mechanism on a continuous type space.
 This is similar in motivation to the transformation in the present paper, where we leverage 
Lipschitz continuity of a trained network in extending 
 menu compatibility on discrete grid points to the entire continuous space.
%
Inspired by~\citet{hartline2011bayesian}, \citet{daskalakis2012symmetries} use a two step process to achieve
exact BIC,
 with an agent type first bidding against copies of itself (drawn from the same type distribution as the agent) 
in a VCG auction to choose the ideal surrogate from those with types  sampled on the grid,
%
with the optimized mechanism then running on
 the chosen surrogates, awarding each agent the allocation and payment of its surrogate.
 This method and proof relies heavily on distributional analysis specific to the Bayesian setting and is unlike 
 our menu-based,  Lipschitz-reasoned transformation.
 A more naïve approach \citep{chawla2007algorithmic}, where each agent's reported values are rounded to the nearest grid points, with the prices discounted to preserve the incentive scheme, works in the single bidder case, but as \citet{daskalakis2012symmetries} remark,  does not generalize to the multi-bidder setting. See also \citet{conitzer2021welfare} and~\citet{cai2021efficient} for  $\epsilon$-BIC  to  BIC 
 transforms with different kinds of theoretical guarantees, including in~\citet{conitzer2021welfare} a transform to BIC for RegretNet in the special case of uniform, independent types. It is worth highlighting that these transforms only work for BIC problems and do not extend to the case of DSIC. To our knowledge, reducing $\epsilon$-expected-expost-IC ($\epsilon$-EEIC, which is what RegretNet converges to) to DSIC remains an open question. \citet{conitzer2021welfare} demonstrate that $\epsilon$-EEIC (Theorem 11) can be reduced to BIC   for independent uniform type distributions. But we know of no result for an $\epsilon$-EEIC to DSIC reduction. The surrogate matching technique discussed in~\citet{daskalakis2012symmetries} is limited to BIC and doesn't provide an $\epsilon$-IC (or $\epsilon$-EEIC) to DSIC reduction (see Appendix K after Lemma 15). See also~\citet{cai2021efficient}.

Other work in automated mechanism design aims to optimize parametric
designs within a restricted family
 of known SP  mechanisms~\citep{guo2010computationally, sandholm2015automated}, e.g., affine maximizer auctions. As with the differentiable approach in \citet{duan2023scalable} and \citet{curry2022differentiable}, this
 restricted search space is not fully expressive.

\citet{guo2024worst} integrates neural networks with mixed integer programming (MIP) in a study of worst-case VCG redistribution mechanism design for the public project problem. He proposes using a neural network to model the payment function, and makes use of MIP to solve for worst-case type profiles that maximally violate mechanism design constraints. In contrast to our approach, where the MIP is used solely to derive price adjustments ensuring feasibility, \citet{guo2024worst} uses these worst-case type profiles as training samples to adapt 
towards better worst-case mechanisms.
 
 Using machine learning approaches, \citet{narasimhan2016automated} consider
 automated design in domains without money, and \citet{dutting2015payment} leverage max-margin  methods 
 to learn an optimally incentive-aligned
 payment rule given an allocation rule (but without finding exactly SP mechanisms except for implementable allocation rules).
This paper also sits within the growing body of work integrating machine learning with economic theory and practice, e.g., \citet{hartford2016deep, kuo2020proportionnet, ravindranath2023data, wang2023deep, ravindranath2021deep, tacchetti2019neural, peri2021preferencenet, rahme2021permutation, wang2024multi}. 
\section{Preliminaries}

\textbf{Sealed-Bid Auction}. We consider sealed-bid auctions with a set of $n$ bidders, $N=\{1,\ldots,n\}$, and $m$ items, $M=\{1,\ldots,m\}$, 
where bidder $i$ 
has a {\em valuation function}, $v_i: 2^M \rightarrow \mathbb{R}_{\ge 0}$. 
We consider two kinds of valuation functions. For bidders with \emph{additive valuations}, the value for a subset of items $S\subseteq M$ is $v_i(S)=\sum_{j\in S}v_i({j})$, where $v_i({j})$ is the value for an item $j\in M$. Alternatively, bidders can have \emph{unit-demand valuations}, where $v_i(S)=\max_{j\in S}v_i({j})$ for some $S\subseteq M$.
Valuation $v_i$ is drawn independently from a distribution $F_i$ defined on the space of possible valuation functions $V_i$. We consider bounded valuation functions: $v_i(j)\in[0, v_{\max}]$, for $\forall i$, $j$, with $v_{\max}>0$.
%
We use $\vv_i=(v_i(1), \cdots, v_i(m))$ to denote the value of bidder $i$ for 
each of the $m$ items.

The auctioneer is assumed to know the distribution $\bm F=(F_1,\cdots,F_n)$, but not the realized valuation profile $\vv=(v_1,\cdots,v_n)\in V$. The bidders report their valuations, perhaps untruthfully, as their {\em bids},
 $\bm{b}=(b_1,\cdots,b_n)$, where $b_i\in V_i$.
The goal is to design an auction $(g,p)$ to maximize expected revenue.
This is an {\em allocation rule}, $g: V\rightarrow \mathcal{X}$,  where
$\mathcal{X}$ is the space of feasible allocations (i.e., no item allocated more than once),
and 
a {\em payment rule}, $p_i: V\rightarrow \mathbb{R}_{\ge 0}$, to each bidder $i$. 
We also write $g_i(\bm b)\subseteq M$
 to denote the set of items (perhaps empty) allocated to bidder $i$ at bid profile $\bm b$.
The utility to bidder $i$ with valuation function $v_i$ at bid profile $\bm b$ is $u_i(v_i;\bm b)=v_i(g_i(\bm b))-p_i(\bm b)$.
 In full generality, the allocation and payment rules may be randomized, with
 each bidder assumed to be risk neutral  and seeking
to maximize its 
expected utility.

In a \emph{dominant-strategy incentive compatible} (DSIC) auction, or {\em strategy-proof (SP)} auction, each bidder's utility is maximized by bidding its true valuation $v_i$ regardless of other bids; i.e., $u_i(v_i;(v_i, \bm b_{\shortn i}))\ge u_i(v_i;(b_i, \bm b_{\shortn i}))$, for $\forall i\in N, v_i\in V_i, b_i\in V_i$, and $\bm b_{\shortn i}\in V_{\shortn i} $.
An auction is \emph{individually rational} (IR) if each bidder receives a non-negative utility when participating and truthfully reporting: $u_i(v_i;(v_i, \bm b_{\shortn i}))\ge 0$, for $\forall i\in N, v_i\in V_i$, and $\bm b_{\shortn i}\in V_{\shortn i}$.
Following the revelation principle, it is without loss of generality to focus on direct, SP
auctions, as any auction that achieves a particular expected revenue in a dominant-strategy equilibrium
 can be transformed into an SP auction with the same revenue.
 Optimal auction design therefore seeks to identify a SP and IR auction that maximizes the expected revenue, i.e., $\mathbb{E}_{\vv\sim \bm F}[\sum_i p_i(\vv)]$. 

\textbf{Menu-Based Auction Design}. In {\em menu-based auction design}, the allocation and payment rules for 
each bidder $i$ are 
represented through a {\em menu function}, $B_i(\bm b)$, that 
generates a menu to bidder $i$ consisting of
$K\ge 1$  {\em menu elements} (in general, $K$ may vary by bidder $i$ and
input $\bm b$)
 for every possible bid profile $\bm b$. 
We write $B_i(\bm b)=(B_i^{(1)},\ldots,B_i^{(K)})$, 
and the  $k$th \emph{menu element}, $B_i^{(k)}$,
 specifies a {\em bundle},
 $\bm\alpha^{(k)}_i\in [0,1]^m$, and a {\em price}, $\beta^{(k)}_i\in \mathbb{R}$,
 to bidder $i$.
Here, we allow randomization, where  $\alpha^{(k)}_{ij}\in[0,1]$  denotes the
  probability that item $j$ is assigned to agent $i$ in 
  menu element $k$. 
  %
   %
  %
We refer to menu functions $B=(B_1,\ldots,B_n)$ as corresponding to a {\em menu-based representation 
of an auction.}
  The following theorem gives a necessary 
 condition for a menu-based 
  representation of an auction  to be SP. Let $b_i(\bm\alpha^{(k)}_i)$ denote the expected value (as represented by the bid $b_i$ of bidder $i$) for the possibly
randomized bundle $\alpha^{(k)}_i$.
\begin{theorem}[SP auctions via menus (necessary)~\citep{hammond1979straightforward}]\label{thm:menu}
    An auction $(g, p)$ is SP only if there is a menu-based representation, $B=(B_1,\ldots,B_n)$,  that satisfies, 
    
    1. {\em (Bidder Optimizing)} For every bidder $i$, and  every bid profile, $\bm b$,
    \begin{align}
        (g_i(\bm b), p_i(\bm b)) \in \argmax_{\left(\bm\alpha^{(k)}_i, \beta^{(k)}_i\right)\in B_i(\bm b)}[b_i(\bm\alpha^{(k)}_i)-\beta^{(k)}_i]; \ \mbox{and}
    \end{align}
    
    2. {\em (Self-Bid Independent)} For each bidder $i$, the menu function, $B_i(\bm b)$, is  independent of their own bid, $b_i$,
    that is $B_i(b_i,\bm b_{\shortn i})= B_i(b'_i,\bm b_{\shortn i})$, for all $b_i, b'_i$, and
     $\bm b_{\shortn i}$.
\end{theorem}

Given a menu-based representation, let $\bm\alpha^{*}_i(\bm b)$ and $\beta^{*}_i(\bm b)$ denote the
bundle and price components of the  bidder-optimizing menu element 
for bidder $i$   at  bid profile $\bm b$. In cases where the context is unambiguous, we omit the explicit dependence on $\vb$ and write $\bm\alpha^{*}_i$ and $\beta^{*}_i$.
%
We now define  {\em menu compatibility},  which  is required for a menu-based representation
to provide a SP auction.
\begin{definition}[Menu compatibility]
The menus available to each bidder are {\em  menu compatible at bid profile $\bm b$}  if the  bidder-optimizing choices, $\bm\alpha^{*}(\bm b)=(\bm\alpha^{*}_1(\bm b),\ldots,\bm\alpha^{*}_n(\bm b))$, are free from over-allocation, so that $\sum_i \bm\alpha^{*}_i(\bm b)\le 1$.
  A menu-based representation $B=(B_1,\ldots,B_n)$ 
  is {\em menu compatible} if the menus are menu compatible 
   for every bid profile $\bm b\in V$.
\end{definition}

In words, menu compatibility requires that the bidder-optimizing choices for each bidder from its menu,
at each input to the auction, are always feasible when considered together, in that no item is over-allocated.
\begin{theorem}[SP auctions via menus (sufficient)~\citep{hammond1979straightforward}]\label{thm:menu2}
   An auction defined   through a menu-based representation, $B=(B_1,\ldots,B_n)$, is SP 
  if the menu-based representation is self-bid independent, choices are bidder optimizing,
  and these choices satisfy menu compatibility.
\end{theorem}

In full generality, one needs to also handle tie-breaking in defining menu compatibility (i.e., 
there may be ties in bidder-optimizing choices, and it is sufficient
that there always exists a way to break ties across equally good, bidder-optimizing choices, so as to not to over-allocate).
We do not need this because our computational framework,  in the construction of the MILP for price adjustments,
ensures that the best menu choices to bidders  are (1) menu compatible, and (2) 
for each bidder, have a utility 
that is {\em strictly} larger than that of the second best choice with a safety margin that
depends on the Lipschitz constant of the trained network
and the distance between grid points.

In the context of menu-based auction design, the 
optimal auction design problem is to identify a menu-based representation that 
maximizes  expected revenue, i.e., $\mathbb{E}_{\vv\sim \bm F}[\sum_i \beta^{*}_i(\vv)]$
amongst all self-bid independent and menu compatible 
representations.

\section{The General Menu-Based NETwork and Price Transformation Method}

In this section, we introduce our method for how to learn revenue-maximizing, SP, and multi-bidder auctions. Our method comprises two key components. Sec.~\ref{sec:method:dl} details the training of neural networks for self-bid independent, menu-based representations, where we introduce an {\em incompatibility loss function} that encourages  menu compatibility.
In practice, we find that this may leave  a small over-allocation rate and thus a small failure of menu compatibility.
 To address this issue and ensure menu compatibility on the entire value domain, Sec.~\ref{sec:method:t} introduces a
novel  {\em menu transformation} technique. In Sec.~\ref{sec:thoery}, we prove that the auction design after transformation 
is menu compatible on the entire input domain, and thus exact SP.

\subsection{Deep Menu Learning}
\label{sec:method:dl}


We train a neural network $B(\theta)$ with parameterization $\theta$ to learn a menu-based representation that maximizes revenue while also driving 
down the rate of menu incompatibility.
 Specifically, $B(\theta)$ consists of a pair of neural network components, $(f_{\xi_i}, q_{\zeta_i})$, 
for each bidder $i$, where these
are parameterized by $\theta_i=(\xi_i, \zeta_i)$.
%
$f_{\xi_i}$ is the {\em bundle network} for bidder $i$ and 
generates the bundles associated with each menu element,
 and $q_{\zeta_i}$ is the {\em price network} for bidder $i$
 and generates the prices associated with 
each menu element.
 To meet the requirement of SP (Theorem~\ref{thm:menu2}), the inputs to $f_{\xi_i}$ and $q_{\zeta_i}$ 
 only depend on the bid profile,  $\vb_{\shortn i}$, of the other bidders.
$f_{\xi_i}$ generates a set of $K- 1$ bundles, $\bm\alpha_i(\vb_{\shortn i};\xi_i)=f_{\xi_i}(\vb_{\shortn i})\in[0,1]^{(K\shortn 1)\times m}$,
and  $q_{\zeta_i}$  a set of $K-1$ prices, $\bm\beta_i(\vb_{\shortn i};\zeta_i)=q_{\zeta_i}(\vb_{\shortn i})\in\mathbb{R}^{K-1}$. The $k$-th row of the output of the bundle network  and the $k$-th element of the output of the  price network  make up the $k$-th element in the menu for bidder $i$, 
$(\bm\alpha_i^{(k)}(\vb_{\shortn i}), \beta_i^{(k)}(\vb_{\shortn i}))$. Appendix~\ref{appx: nntraining} gives a detailed description and visualization of the network architecture.
 Hereafter we omit the network parameters for simplicity. To ensure individual rationality (IR), we also include a fixed $K$-th element, with $\bm\alpha_i^{(K)}(\vb_{\shortn i})=\bm 0$, $\beta_i^{(K)}(\vb_{\shortn i})=\bm 0$. 
 
 During training we assume truthful inputs since we will attain a SP mechanism, and minimize:
\begin{align}
    \mathcal{L}(\theta) = -\mathcal{L}_{\textsc{Rev}}(\theta) + \lambda_{\textsc{Incomp}}\cdot 
\mathcal{L}_{\textsc{Incomp}}(\theta)\label{equ:loss}.
\end{align}
The first term in this loss function is for maximizing the empirical revenue, i.e., 
\begin{align}
    \mathcal{L}_{\textsc{Rev}}(\theta) = \frac{1}{|D|}\sum_{\vv\in D}\left[\sum_{i\in N}\sum_{k\in[K]}\vz_i^{(k)}(\vv)\beta_i^{(k)}(\vv_{\shortn i}) \right],
\end{align}
where $D$ is a set of bidders' values sampled from $\bm F$ and $\vz_i^{(k)}(\vv)$ is obtained by applying the differentiable
SoftMax function to the utility of bidder $i$ being allocated the $k$-th menu choice, i.e.,
\begin{align}
    \vz_i^{(k)}(\vv) = \mathsf{SoftMax}_k\left(\lambda_{\textsc{SoftMax}}\cdot u^{(1)}_i(\vv),\ldots,\lambda_{\textsc{SoftMax}}\cdot u^{(K)}_i(\vv)\right).\label{equ:softmax_in_loss}
\end{align}

For additive valuations, the utility $u^{(k)}_i(\vv) = \vv_i^\Tau\bm\alpha_i^{(k)}(\vv_{\shortn i}) - \beta_i^{(k)}(\vv_{\shortn i})$. This also applies to unit-demand bidders with the constraint $\sum_j\alpha_{ij}^{(k)}=1$. Also,
$\lambda_{\textsc{SoftMax}}>0$ is a scaling factor that controls the quality of the approximation. 

\textbf{Incompatibility loss}. Until this point, the network formulation
generalizes RochetNet~\citep{duetting2023optimal},
 allowing the menu  of each 
 bidder to depend on the bids of others.
 The new challenge  in the multi-bidder setting is that bidder optimizing choices may be  incompatible, with one or more items needing to be over-allocated, when taking together the choices of each bidder.
%
To address this issue,  the  second term in the loss function (Eq.~\ref{equ:loss}) is 
the {\em incompatibility loss}:
%
\begin{align}
    \mathcal{L}_{\textsc{Incomp}}(\theta)=\frac{1}{|D|}\sum_{\vv\in D}\left[\mathsf{ReLU}\left(\sum_{i\in N, k\in[K]}\vz_i^{(k)}(\vv)\bm\alpha_i^{(k)}(\vv_{\shortn i}) - (1 - s_f)\right)\right],\label{equ:f_loss}
\end{align}
where   $s_f>0$ is a 
{\em safety margin},
and
this loss term is associated with scaling factor, $\lambda_{\textsc{Incomp}}>0$, in Eq.~\ref{equ:loss}, balancing revenue and menu incompatibility.
This incompatibility loss sums up the SoftMax-weighted bundles selected by each bidder and encourages compatible choices, with a positive  loss
only when the summed allocation is larger than $1-s_f$ for one or more items. The loss (Eq.~\ref{equ:loss}) follows the commonly used structure of Lagrangian-based optimization methods, where the incompatibility term acts as a Lagrange multiplier that encourages feasibility.

%

%

\textbf{Safety margin}. 
 %
We achieve full menu compatibility over the domain through the
 menu transformation technique   introduced in Sec.~\ref{sec:method:t}. 
 There, we setup a series of MILPs to enforce menu compatibility on a grid of bidder values.
 With the {\em Lipschitz constant} of the bundle network, $L_a>0$,
 and the interval (in $\ell_\infty$-norm) between two grid points used in the price adjustment, $\epsilon>0$, we prove in Sec.~\ref{sec:thoery} that a safety margin $s_f = n\cdot\epsilon\cdot L_a/2$ 
(where $n$ is the number of bidders and $s_f$ is also used in the construction of the MILP)
allows  menu compatibility on the grid to extend to the entire continuous value domain.
%
%
A technical challenge is that a large safety margin 
 can lead to reduced  revenue. To limit the size of the safety margin,
 we constrain the Lipschitz constant of the bundle network using
 {\em spectral normalization}~\citep{miyato2018spectral}.
 This technique, when applied to fully connected neural networks, involves dividing the weight matrix of each layer by its largest singular value. In our experiments, the Lipschitz constant $L_a$ 
 is often in range $1e\shortn 5$ to  $1e\shortn 4$ following spectral normalization. 
 Sec.~\ref{sec:exp_setup} gives more
 information about our network architecture and training schemes.


\subsection{Price Adjustment as Menu Transformation}\label{sec:method:t}

The trained bundle and price networks generate for each bidder $i$ a menu $B_i(\vv_{\shortn i}; \theta_i)=(\bm\alpha_i(\vv_{\shortn i}), \bm\beta_i(\vv_{\shortn i}))$ given $\vv_{\shortn i}$. Empirically, we find that these menus suffer from a small rate of incompatibility (i.e., item over-allocation). 
To address this problem, we introduce a  menu transformation technique, applied after training but before using the trained mechanism,
to modify the prices in menus and
ensure menu compatibility.

The menu transformation  adjusts the prices  in a menu. 
%
 For this, we generate a grid of bidders' values $\mathcal{V}=\mathcal{V}_i\times\mathcal{V}_{\shortn i}$, which we use  to define a series of mixed integer linear programs (MILPs). These programs are designed to determine  price changes that
 ensure bidder-optimizing choices are compatible for all grid points.
 We prove in Sec.~\ref{sec:thoery}
 that  menu compatibility on this discrete grid  extends to the entire continuous value domain.
The transformed mechanism is  defined on the entire continuous value domain.
For a general input $\vv'$ that  may not be on the grid and some bidder $i$,
 we use the outputs from
the trained bundle and price networks at $\vv'_{\shortn i}$ to obtain a pre-transformed menu for $i$,
and then 
adjust the  menu using the price adjustments from the closest grid point in terms of $\ell_\infty$-norm, i.e., $\forall \vv'\in\{\vv''\,\, |\,\, \|\vv''-\vv\|_\infty\le \epsilon/2\}$ uses the price adjustments at grid point $\vv$.
%
 
We first illustrate the method in the two-bidder case. For each $\vv_{\shortn i}\in\mathcal{V}_{\shortn i}$ , we construct a separate MILP, where the decision variables are adjustments to the prices $\Delta\beta^{(k)}_i, k\in[K\shortn 1]$, 
 the constraints are to ensure that for any $\vv_i\in\mathcal{V}_i$, the two bidders will  make compatible choices, and  the objective is to minimize the sum absolute price adjustment.
%
%
As $\vv_i$ varies, the utility-optimizing choice of $i$ may vary, 
and so too may the menu of bidder $\shortn i$ and thus the choice of bidder $\shortn i$.
We achieve  compatibility between the choices by bidder $i$ at grid $\mathcal V_i$ and bidder $\shortn i$ at value $\vv_{\shortn i}$.

Let $\mathcal{V}_i=\{\vv_{i,(\ell)}\}_{\ell=1}^{|\mathcal{V}_i|}$. For each $\vv_{i,(\ell)}\in \mathcal{V}_i$ (all $\vv_{i,(\ell)}$ share the same menu $B_i(\vv_{\shortn i})$ as $\vv_{\shortn i}$ is fixed), there are two kinds of menu elements: 

(i) $k\in B_{i,(\ell)}^{\mathsf{comp}}(\vv_{\shortn i})$, which means the $k$-th menu element is compatible with the menu element selected by bidder $\shortn i$ given $\vv_{i,(\ell)}$: 
\begin{align}
    \bm\alpha^*_{\shortn i}(\vv_{i,(\ell)}) + \bm\alpha_i^{(k)}(\vv_{\shortn i}) \le 1-s_f;\label{equ:mip_feas_standard}
\end{align}

(ii) $k\in B_{i,(\ell)}^{\mathsf{incomp}}(\vv_{\shortn i})$, which means the $k$-th menu element is incompatible with the choice of bidder $-i$. 

Safety margin $s_f$ is used to define these two kinds of menu elements, which
affects the variables in the MILP.
We use the following MILP to adjust prices associated with bidder $i$'s menu elements: 
%
%
\begin{align}
\text{Decision} &  \text{ variables: price adjustments} \{\Delta\beta_i^{(k)}\}_{k=1}^{K-1}\text{;} \nonumber \\
\text{Binary v} & \text{ariables: } z^{(\ell k)}, \text{for } l\in[|\mathcal{V}_i|],k\in B_{i,(\ell)}^{\mathsf{comp}}(\vv_{\shortn i}); \text{Variables: } U_{(\ell)};\nonumber\\
    \min& \sum_{k\in[K-1]} |\Delta\beta_i^{(k)}|\tag{Objective}\label{eq:trans_obj}\\
s.t.\ \text{For } &\forall\ell, \forall k\in B_{i,(\ell)}^{\mathsf{comp}}(\vv_{\shortn i}):\tag{Constraint Set 1}\label{eq:trans_feas}\\
& U_{(\ell)} \ge \vv_{i,(\ell)}^\Tau\bm\alpha_i^{(k)} - \beta_i^{(k)} - \Delta\beta^{(k)}_i + (1-z^{(\ell k)})s_m, \label{eq:sm_for_comp_element}\\
& U_{(\ell)} \le \vv_{i,(\ell)}^\Tau\bm\alpha_i^{(k)} - \beta_i^{(k)} - \Delta\beta^{(k)}_i + (1-z^{(\ell k)})M,\nonumber\\
& z^{(\ell k)}\in\{0,1\}, \sum\nolimits_{k\in B_{i,(\ell)}^{\mathsf{comp}}(\vv_{\shortn i})} z^{(\ell k)} = 1,\nonumber \\
\text{For } & \forall\ell, \forall k'\in B_{i,(\ell)}^{\mathsf{incomp}}(\vv_{\shortn i}):\tag{Constraint Set 2}\label{eq:trans_infeas}\\
& U_{(\ell)} \ge \vv_{i,(\ell)}^\Tau\bm\alpha_i^{(k')} - \beta_i^{(k')} -\Delta\beta^{(k')}_i + s_m.\label{eq:sm_for_incomp_element}
\end{align}

As we will explain, $s_m>0$ and $M>0$ are set to be suitable constants and, in  full generality, the price changes may be positive or negative.
This MILP ensures that the best compatible choice for $i$ at each possible value $\vv_i$, and fixing some $\vv_{\shortn i}$,
 is more appealing than any incompatible choice.
 In (\ref{eq:trans_feas}), we use the ``big-M method"
 and introduce binary variables $z^{(\ell k)}\in \{0,1\}$ 
to identify the maximum utility $U_{(\ell)}$ achievable by compatible elements. $z^{(\ell k)}$ is 1 only when the $k$-th menu element is the best choice of bidder $i$ of type $\vv_{i, (\ell)}$. $M$ is set to be sufficiently large (see Appendix~\ref{appx:big_M} for the detail). 
In (\ref{eq:trans_infeas}), we enforce that $U_{(\ell)}$ is larger than the utility of any incompatible element by a second {\em safety margin}, $s_m=L_a(mv_{\max}\epsilon+m\epsilon^2/2)+\epsilon L_p+\epsilon$, where $\epsilon$ is the interval (in $\ell_\infty$-norm) between two grid points in $\mathcal{V}$, $L_a$ the Lipschitz constant of the bundle network, and $L_p$  the Lipschitz constant of the price network. 
%
Similarly, we introduce $s_m>0$ in Eq.~\ref{eq:sm_for_comp_element}, so that the utility gap between the best and the second best compatible choice is at least $s_m$. 
%
Together, Eq.~\ref{eq:sm_for_comp_element} and~\ref{eq:sm_for_incomp_element} ensures that the utility of the best menu choice is larger than the others by at least $s_m$. 
We will see in Sec.~\ref{sec:thoery} that this safety margin allows us to provably establish menu compatibility, though this transformation, on 
the entire continuous value domain. This also removes any tie-breaking issue
 in making bidder-optimizing choices.
 %


The objective   of the MILP is to minimize the sum absolute price change.
%
We include the IR ($K$-th) element when finding the maximum achievable utility of feasible elements in (\ref{eq:trans_feas}), as this is always a compatible choice regardless of the selection of bidder $\shortn i$,
and fix the price of this element to 0 during transformation. 
For the two-bidder case, it is sufficient to guarantee menu compatibility by only adjusting the menu prices of bidder 1.

For the case of two or more bidders, the  price adjustment process proceeds for each bidder in  increasing order of bidder index. 
For bidder $i$, we again consider different possible values $\vv_{\shortn i}$ on a grid,
and consider different
values $v_i\in \mathcal V_i$ in constructing an MILP for each $\vv_{\shortn i}$.
 This MILP is the same as in the two-bidder case after identifying compatible menu choices. Considering the already transformed menus for the preceding bidders in the transform order, the $k$-th element of bidder $i$ is compatible if:
\begin{align}
    \bm\alpha_i^{(k)}(\vb_{\shortn i}) + \min\left(1- s_f, \sum_{j=1}^{i-1}\tilde{\bm\alpha}^*_{j}(\vb_{\shortn j})+\sum_{j=i+1}^n\bm\alpha^*_{j}(\vb_{\shortn j}) \right) \le 1- s_f,\label{equ:multi_mip_feas_standard}
\end{align}
%
where $\tilde{\bm\alpha}^*_{j}(\vb_{\shortn j})$ is the optimal choice for bidder $j< i$ given its already transformed menu. We clip the aggregate allocation of other bidders to $1-s_f$ because even a 0 allocation is infeasible if this aggregate allocation of other bidders is larger than $1- s_f$.
In Theorem~\ref{thm:general_comp}, we prove that these MILP formulations,
when solved for all bidders and considering all grid points of other bidders, ensures menu compatibility on the entire, continuous value domain.
 \begin{table}[t]
    \caption{Reducing MILP complexity for the 2-bidder, 2-item, i.i.d. uniform values on $[0,1]$ setting. $100\times 100$ grid points are used for each bidder except when using the adaptive grid strategy. Strategies are incrementally added row by row.\label{tab:eff_strategy}}
    \centering
    \begin{tabular}{CRCRCRCRCR}
        \toprule
        \multicolumn{2}{l}{\multirow{2}{*}{Strategy}} &
        \multicolumn{2}{c}{\multirow{2}{*}{\# MILP}} &
        \multicolumn{6}{c}{Per MILP Complexity (Mean $\pm$ Var)}\\
        
        \cmidrule(lr){5-6}
        \cmidrule(lr){7-8}
        \cmidrule(lr){9-10}
        \multicolumn{2}{c}{} &
        \multicolumn{2}{l}{} &
        \multicolumn{2}{c}{\# Constraints} & 
        \multicolumn{2}{c}{\# Binary Variables} & 
        \multicolumn{2}{c}{Run Time (s)}\\
        \midrule
        \multicolumn{2}{l}{Nothing} & \multicolumn{2}{c}{10,000        } & \multicolumn{2}{c}{3,568,866$\pm$271,276} & \multicolumn{2}{c}{548,866$\pm$271,276} & \multicolumn{2}{c}{9771$\pm$11,985}  \\
        \multicolumn{2}{l}{+ 0-Violation} & \multicolumn{2}{c}{3527} & \multicolumn{2}{c}{3,568,866$\pm$271,276} & \multicolumn{2}{c}{548,866$\pm$271,276} & \multicolumn{2}{c}{9771$\pm$11,985}  \\
        \multicolumn{2}{l}{+ Keep-Choice} & \multicolumn{2}{c}{3527          } & \multicolumn{2}{c}{2,550,387$\pm$244,295} & \multicolumn{2}{c}{38,312$\pm$18,647  } & \multicolumn{2}{c}{629.3 $\pm$154.3}  \\
        \multicolumn{2}{l}{+ IR-Screening} & \multicolumn{2}{c}{3527          } & \multicolumn{2}{c}{1491$\pm$1011        } & \multicolumn{2}{c}{373$\pm$223        } & \multicolumn{2}{c}{2.93$\pm$1.52  }  \\
        \multicolumn{2}{l}{+ Adaptive Grid} & \multicolumn{2}{c}{3527          } & \multicolumn{2}{c}{750$\pm$456          } & \multicolumn{2}{c}{28$\pm$188.6      } & \multicolumn{2}{c}{1.09$\pm$0.44  }  \\
        \midrule
        \multicolumn{2}{l}{Overall Reduction} & \multicolumn{2}{c}{-64.63\%} & \multicolumn{2}{c}{-99.98\%} & \multicolumn{2}{c}{-99.995\%} & \multicolumn{2}{c}{-99.99\%}\\
        \bottomrule
    \end{tabular}
\end{table}
\subsection{Accelerating the MIP}\label{sec:method:rdc}

The challenge with solving this series of MILPs lies in the computational efficiency.
In its basic form, for each $\vv_{i,(\ell)}\in \mathcal{V}_i$, a binary variable is assigned to each feasible menu element, resulting in a total of $\sum_{\ell\in[|\mathcal{V}_i|]}|B_{i,(\ell)}^{\mathsf{comp}}(\vv_{\shortn i})|$ binary variables in the MILP to transform $i$'s menu for $\vv_{\shortn i}$.
 Besides, the number of MILPs grows exponentially with the number of bidders and items because we need a MILP for each grid point $\vv_{\shortn i}\in\mathcal{V}_{\shortn i}$. 

Fortunately, our deep learning process already provides a menu representation that is Lipschitz smooth and almost  menu compatible, enabling us to develop the following
strategies to greatly reduce the running time without sacrificing 
the  guarantee on menu compatibility on the full domain. We show in Table~\ref{tab:eff_strategy} that 
these strategies allow us to
harness our computational pipeline to efficiently find 
SP, highly revenue-optimizing auctions.


\textbf{(I) 0-Violation}. For those grid points $\vv_{\shortn i}\in\mathcal{V}_{\shortn i}$ where the menus are already compatible, we do not need to run the MILP. This strategy reduces the number of MILPs to solve.

\textbf{(II) Keep-Choice}. In each MILP, for most $\vv_{i,(\ell)}$, the learned networks 
provide a feasible joint allocation, and we can seek to retain the  choice of bidder $i$ in this case during
price transformation. These \emph{keep-choice} (\namekp) constraints can be described by linear expressions without introducing a binary variable:
\begin{equation}
    \vv_{i,(\ell)}^\Tau\bm\alpha_i^{(k)} - \beta_i^{(k)}-\Delta\beta_i^{(k)}\le \vv_{i,(\ell)}^\Tau\bm\alpha_i^* - \beta_i^*- \Delta\beta_i^*, \ \ \forall k\in[K],\tag{KC}
\end{equation}
where $(\bm\alpha_i^*, \beta_i^*)$ is the optimal choice for $\vv_{i,(\ell)}$ in the network-generated menus. 
The critical consideration here is the extent to which we can
 apply \namekp~constraints. Applying \namekp~indiscriminately to all compatible $\vv_{i,(\ell)}$ might render the MILP at $\vv_{\shortn i}$ infeasible. 
This is because addressing other values of $\vv_i$ with over-allocation  
may require  price adjustments that lead to changes in the agent's choice at $\vv_{i,(\ell)}$.
 We adopt a  heuristic approach that works well empirically:
 sort compatible $\vv_{i,(\ell)}$ by  ascending difference between the utilities of their best and second-best menu elements, and apply the \namekp~constraint to the top $c\%$ compatible $\vv_{i,(\ell)}$, for some choice of $c\geq 0$.
 This heuristic is grounded in the observation that our price adjustment process changes the prices only
 slightly and bidders tend to stick with their initial choices when the alternatives are significantly less favorable.  In Table~\ref{tab:eff_strategy}, we choose $c=95$,
 and find  this strategy reduces the number of binary variables in each MILP by $93.02\%$ on average.  If this $c\geq 0$ value makes an MILP infeasible, we iteratively reduce $c$ value by 5 until the MILP becomes feasible, noting that the MILP is always feasible when $c=0$, as all \namekp~constraints are removed.

\textbf{(III) IR-Screening}. For those $\vv_{i,(\ell)}$ to which we introduce binary variables and apply the big-M method, we do not need a binary variable for each compatible menu element. We find that most (typically $>99\%$) of the learned menu elements have a negative utility. By adding a constraint that all 
 price adjustments, $\Delta\beta_i^{(k)}$, are non-negative,
the utility of these menu elements remain negative and they are not selected over the 
 null (utility 0) choice that ensures IR.
 In this way, we can safely remove the corresponding binary variables and constraints for these elements.  In Table~\ref{tab:eff_strategy}, this strategy further removes $89.71\%$ of the binary variables and $99.94\%$ of the constraints on average.\footnote{Allowing negative price changes has a  
 small positive influence on  revenue, but we find empirically that this effect is very limited (e.g., from 5.0207 to 5.0271 in the 3-bidder, 2-item setting with the irregular valuation distribution ($\prescript{3}{2}{IRR}_\mathtt{Add}$)). One explanation is that our transform tends to address over-allocation, where increasing prices can help, e.g., by encouraging the null allocation option. Therefore, we recommend adopting non-negative price adjustments along with IR-screening for a good revenue-computational cost tradeoff.}
 
\textbf{(IV) Adaptive Grid}. Our theoretical analyses in Sec.~\ref{sec:thoery} provides 
 for an opportunity to adaptively decrease the size of the grid. Appendix~\ref{sec:exp_adaptive_grid} describes this {\em adaptive grid} strategy in detail.


\subsection{Discussion}

Our method makes use of MILPs to ensure the compatibility of learned menus on the entire value domain. An interesting question is 
whether this approach  may also open up 
a new, purely MILP-based methodology. To the best our knowledge, there is no previous work on computing revenue-maximizing menus from scratch with MILPs and
exploring new MILP-only methods, with generalization away from a grid and the new ideas
introduced here of safety margins, together with menu-based representations, seems possible and interesting
to explore in future work.
 For example, starting from our MILPs, one could additionally 
 introduce decision variables for menu bundles.
Although worth exploring, we identify three challenges stemming from the absence of deep learning initialization and hope to spark further investigation with this brief discussion. 

(1) The lack of initialization through deep learning would preclude sequentially changing the menus of individual bidders while keeping the menus of others fixed, because it is hard to distinguish compatible and incompatible menu elements as in Eq.~\ref{equ:mip_feas_standard} when constructing MILPs. Rather, a pure MILP approach would naively need to consider the menus of all bidders simultaneously, within a single MILP. This MILP would require an
 additional $GKmn$ decision variables to model bundles, where the grid size 
$G$ increases exponentially with the number of items $m$ and the number of bidders $n$, and $K$ is the menu size.  It would also necessitate an additional $O(GKn)$ binary variables and $O(Gm)$ bilinear (as they involve multiplication of bundle variables with binary variables) menu-compatibility constraints using the big-M method. 

(2) Furthermore, techniques in Table 1, which  cut price-adjustment time by >99.99\%, all depend on initialization through deep learning. For example, the 0-violation and keep-choice methods are applicable when menus are already compatible after deep learning, and IR-screening relies on certain menu elements having negative utilities after deep learning. Without deep initialization, MILP-only methods might suffer from prohibitively long running times. 

(3) Another challenge will be determining  how to set the MILP's objective. For example, one could seek to maximize a likelihood-weighted, accumulated price at grid points, 
a bilinear function (as it involves multiplication of price variables with binary variables), but this may result in poor ``off-grid” performance if the grid is sparse. Conversely, with deep learning, we can opt to minimize absolute price changes in price transformation, following learning to maximize expected revenue (while also considering menu incompatibility) during deep learning. This formulation
of the price transformation problem is well suited to revenue goals and empirically proves effective even when using a relatively coarse grid. 
\section{Exact Strategy-Proofness}\label{sec:thoery}



In the price adjustment stage, we use a grid of values to transform the menu. In this section, we establish that the menu compatibility  established on this grid 
 extends to the entire continuous value domain.
The proof leverages the Lipschitz smoothness of the neural networks, which  bounds the changes of the menu across two grid points. 

We begin our analysis with the two-bidder case. With the $\ell_\infty$ distance between two grid points being $\epsilon$, the proof establishes menu compatibility throughout the entire value domain by proving that menu compatibility at a grid point is maintained for all values within an $\ell_\infty$ distance of $\frac{\epsilon}{2}$ from the point.
The following theorem follows from the use of safety margins $s_m$ and $s_f$ in the transformation, which  are set to depend on $\epsilon$ and the Lipschitz constants of
the menu networks. 
%
 %
\begin{theorem}[2-Bidder Menu Compatibility]\label{thm:2_bidder_comp_updated} For the two-bidder case, if $\vv=(v_1,v_2)$ is a grid point used in the menu transformation, then the transformed menu is compatible for any $\vv'\in\mathcal{N}_{\frac{\epsilon}{2}}=\{\vv'\,\, |\,\, \|\vv'-\vv\|_\infty\le \epsilon/2\}$.
\end{theorem}
\begin{proof}
(I) We first prove that neither bidder will  change their menu selection in \neighborjoint. Without loss of generality, we consider bidder $i$'s menu, which is generated conditioned on the bidder $\shortn i$'s value. The utility changes of the $k$-th element in bidder $i$'s menu in \neighborjoint~is affected by two factors: the changes in its allocation and price due to the varying network input $\vv_{\shortn i}$; and the changes in value $\vv_i$. Although the utility also depends on the price adjustments, we use the adjustments at $\vv$ for any  $\vv'\in\mathcal{N}_{\frac{\epsilon}{2}}$, which is a  constant adjustment and does not result in a utility change in \neighborjoint. Specifically, the changes in the allocation and price are bounded when the menu networks are Lipschitz smooth: $\|\tilde\Delta \bm\alpha_i^{(k)}\|_\infty \le L_a \|\Delta \vv_{\shortn i}\|_\infty \le \frac{\epsilon}{2}L_a$, and $\|\tilde\Delta \beta_i^{(k)}\|_\infty \le L_p \|\Delta \vv_{\shortn i}\|_\infty \le \frac{\epsilon}{2}L_p$, and the change $\|\Delta\vv_i\|_\infty$ in $\vv_i$ is bounded by $\epsilon/2$. Therefore, the change in the utility of any menu element $k$ is upper bounded by:
$\|\tilde\Delta u_i^{(k)}\| \le \|\Delta\vv_i^\Tau\bm\alpha_i^{(k)}\|+\|\vv_i^\Tau\tilde\Delta \bm\alpha_i^{(k)}\|+\|\Delta\vv_i^T\tilde\Delta \bm\alpha_i^{(k)}\| + \|\tilde\Delta \beta_i^{(k)}\| \le \|\Delta\vv_i\|_\infty+mv_{\max}\|\tilde\Delta \bm\alpha_i^{(k)}\|_\infty+m\|\tilde\Delta \bm\alpha_i^{(k)}\|_\infty\|\Delta\vv_i\|_\infty + \|\tilde\Delta \beta_i^{(k)}\| \le \frac{\epsilon}{2}+mv_{\max}\frac{\epsilon}{2}L_a+m\frac{\epsilon^2}{4}L_a+\frac{\epsilon}{2}L_p=s_m/2.$

Suppose that the price adjustment to the $k$-th element is $\Delta\beta_i^{(k)}$, and bidder $i$ selects $(\bm\alpha_i^*, \beta_i^*+\Delta\beta_i^*)$ at $\vv_i$. Then, at $\forall \vv'\in\mathcal{N}_{\frac{\epsilon}{2}}$, for any other menu element $k$, we have \begin{align}
    u_i^*(\vv_i') \shortn u_i^{(k)}(\vv_i')\ge  (u_i^*(\vv_i) \shortn \|\tilde\Delta u_i^*\|) \shortn (u_i^{(k)}(\vv_i) \shortp \|\tilde\Delta u_i^{(k)}\|)\ge  u_i^*(\vv_i) \shortn u_i^{(k)}(\vv_i) \shortn s_m.\nonumber
\end{align}

Including the safety margin $s_m$ in Eq.~\ref{eq:sm_for_comp_element} and~\ref{eq:sm_for_incomp_element} ensures that the utility of the best element in the menu is larger than the second best element by at least $s_m$:
$u_i^*(\vv_i) - u_i^{(k)}(\vv_i) \ge s_m.$ It follows that $u_i^{*}(\vv_i') - u_i^{(k)}(\vv_i') \ge 0$, which means the bidder will not change its menu selection.

(II) Since the bidders do not change their selection in \neighborjoint, the only possibility of incompatibility comes from the change in the menu element selected by bidders due to the varying network input. For each bidder $i$, this change is bounded by $\|\tilde\Delta \bm\alpha_{i}^*\|_\infty \le L_a \|\Delta \vv_{\shortn i}\|_\infty\le \frac{\epsilon}{2}L_a.$
    
Here our safety margin $s_f=n\cdot\epsilon\cdot L_a/2$ comes into play. With $s_f$ in the incompatibility loss (Eq.~\ref{equ:f_loss}) and the transformation process (Eq.~\ref{equ:mip_feas_standard}), we have 

\begin{align}
\bm\alpha_i^{(k_i^*)}(\vv') + \bm\alpha_{\shortn i}^{(k_{\shortn i}^*)}(\vv') &\leq  
\bm\alpha_i^{(k_i^*)}(\vv) + \bm\alpha_{\shortn i}^{(k_{\shortn i}^*)}(\vv) 
+ \|\tilde\Delta \bm\alpha_{i}^{(k_i^{*})}\|_\infty + \|\tilde\Delta \bm\alpha_{\shortn i}^{(k_{\shortn i}^{*})}\|_\infty \nonumber \\
&\leq 1 - \epsilon L_a + 2 \cdot \frac{\epsilon}{2}L_a = 1.\nonumber    
\end{align}
Here, $k_i^*$ and $k_{\shortn i}^*$ is the choice of bidder $i$ and $\shortn i$ at $\vv_i$, respectively. This means that the joint allocation is always feasible in \neighborjoint.
\end{proof}

We now extend the discussion to the cases of more than two bidders. 
%
\begin{theorem}[General $n$-Bidder Menu Compatibility]\label{thm:general_comp}
For the $n$-bidder case, if $\vv = (\vv_1, \cdots, \vv_{\shortn n})$ is a grid point used in the
menu transformation, then the transformed menu is compatible for any
$\vv'\in\mathcal{N}_{\frac{\epsilon}{2}}=\{\vv'\,\, |\,\, \|\vv'-\vv\|_\infty\le \epsilon / 2\}$.
\end{theorem}
\begin{proof}
(I) We first prove by contradiction that the joint allocation of bidders at any grid point after transformation
is at most $1- s_f$. Assume that, at a grid point, we have $\sum_{i=1}^{n}\tilde{\alpha}^*_{ij}(\vb_{\shortn i})>1- s_f$ for some item $j$, where $\tilde{\bm\alpha}^*_{i}(\vb_{\shortn i})$ is the allocation for bidder $i$ given the transformed menu. Then there must be a bidder $1<q\le n$ satisfying that $\sum_{i=1}^{q-1}\tilde{\alpha}^*_{ij}(\vb_{\shortn i})\le 1- s_f$ and $\sum_{i=1}^{q}\tilde{\alpha}^*_{ij}(\vb_{\shortn i})> 1- s_f$. 
This is because the allocation to bidder 1 must follow $\tilde{\bm\alpha}^*_{i}(\vb_{\shortn i})\le 1- s_f$. Otherwise, the left-hand side of Eq.~\ref{equ:multi_mip_feas_standard} for bidder 1 is larger than the right-hand side. 

We take a closer look at bidder $q$. Suppose it selects $\tilde{\alpha}_{qj}^*(\vb_{\shortn i})$ after price adjustment. Due to the formulation of our MILPs (Eq.~\ref{equ:multi_mip_feas_standard}), we have
\begin{align}
\tilde{\alpha}_{qj}^*(\vb_{\shortn i}) + \min\left(1 - s_f, \sum_{i=1}^{q-1}\tilde{\alpha}^*_{ij}(\vb_{\shortn i}) 
+ \sum_{i=q+1}^{n}\alpha^*_{ij}(\vb_{\shortn i}) \right) \le 1 - s_f\label{equ:n_proof_comp_q}
\end{align}
We analyze two possible scenarios based on the value of the sum inside the $\min$ function. (1) If $\sum_{i=1}^{q-1}\tilde{\alpha}^*_{ij}(\vb_{\shortn i}) 
+ \sum_{i=q+1}^{n}\alpha^*_{ij}(\vb_{\shortn i}) \ge 1 - s_f$, it follows from Eq.~\ref{equ:n_proof_comp_q} that $\tilde{\alpha}_{qj}^*(\vb_{\shortn i})+1 - s_f\le 1 - s_f$, indicating that $\tilde{\alpha}_{qj}^*(\vb_{\shortn i})=0$. This contradicts the assumption $\sum_{i=1}^{q-1}\tilde{\alpha}^*_{ij}(\vb_{\shortn i})\le 1- s_f$ and $\sum_{i=1}^{q}\tilde{\alpha}^*_{ij}(\vb_{\shortn i})> 1- s_f$. (2) If $\sum_{i=1}^{q-1}\tilde{\alpha}^*_{ij}(\vb_{\shortn i}) 
+ \sum_{i=q+1}^{n}\alpha^*_{ij}(\vb_{\shortn i}) < 1 - s_f$, we have $\sum_{i=1}^{q}\tilde{\alpha}^*_{ij}(\vb_{\shortn i}) 
+ \sum_{i=q+1}^{n}\alpha^*_{ij}(\vb_{\shortn i})<1 -  s_f$, which is also impossible because $\sum_{i=1}^{q}\tilde{\alpha}^*_{ij}(\vb_{\shortn i}) 
+ \sum_{i=q+1}^{n}\alpha^*_{ij}(\vb_{\shortn i})\ge \sum_{i=1}^{q}\tilde{\alpha}^*_{ij}(\vb_{\shortn i}) > 1 - s_f$.



(II) We then follow the case analysis in Theorem~\ref{thm:2_bidder_comp_updated} to prove the menu
 compatibility for any $\vv'\in$\neighborjoint. For case (I), for any menu element of any bidder $i$, the utility change in \neighborjoint~is still upper bounded by $s_m/2$.
 Therefore, our safety margin in Eq.~\ref{eq:sm_for_comp_element} and~\ref{eq:sm_for_incomp_element} can still guarantee that no bidder changes their selections in \neighborjoint. For case (II), the bundle selected by each bidder changes by at most $\epsilon \cdot L_a/2$. As a result, the sum allocation changes by at most $n/2\cdot \epsilon \cdot L_a=s_f/$,  and remains smaller than 1.
\end{proof}

With the guarantee of menu  compatibility, we get all the components required to prove that our 
computational pipeline generates exactly SP auctions.
\begin{theorem}[Exact Strategy-Proofness]\label{thm:sp}
The \name~framework ensures SP auction mechanisms.
\end{theorem}
\begin{proof}
By Theorem \ref{thm:menu2}, a menu-based mechanism is SP 
if (1) menus are self-bid independent, and (2)  choices are bidder optimizing, 
and (3) we have menu compatibility.

(I) The menu networks satisfy  conditions (1) and (2). Bidder $i$'s menu is generated by a network conditioned on the bids of others and is self-bid independent. The mechanism then selects for each bidder  the menu element with the highest utility according to its report. 

(II) The price adjustment process does not violate these two conditions. (1) This step uses a grid 
that is independent of  specific bid  values and moreover, specific bids
 do not exert any influence on the construction of the MILPs. 
Therefore, the price changes as the outcome of this MILP are independent from specific bids.
 (2) This step does not alter the logic of the 
menu-based representation. After price adjustments, each bidder $i$ still faces a menu that's generated (and transformed) independently of its report (misreporting to any $\vv_i'\neq \vv_i$ does not change the menu bidder $i$ receives), 
and bidder $i$ still gets the menu element with the highest utility based on its report (considering adjusted prices).
No bidder can intentionally trigger or avoid the price adjustment process by misreporting.

(III) Menu compatibility is ensured over the entire value domain (Theorem~\ref{thm:general_comp}).
\end{proof}

\begin{table} [t]
    \caption{$\prescript{n}{m}{U}_\mathtt{Add/Unit}$ represents $n$ bidders, $m$ items, and additive or unit-demand valuations, with values  uniformly distributed on $[0,1]$.
    The revenue of \name~consistently exceeds all SP baselines and is close to that of RegretNet, which is not exactly SP. The learned auctions are visualized in Figs.~\ref{fig:mech_compare_2x2},~\ref{fig:u33_compare}, and~\ref{fig:u22_slicing},
    and show an improved allocation structure in \name. See Appendix~\ref{appx:exp_baseline} for the detailed setup of baselines.\label{tab:u01_baselines}}
    \centering
    \begin{tabular}{CRCRCRCRCRCRCR}
        \toprule
        \multicolumn{2}{c}{\multirow{2}{*}{}} &
        \multicolumn{2}{l}{\multirow{2}{*}{Alg.}} &
        \multicolumn{10}{c}{Setting}\\
        
        \cmidrule(lr){5-6}
        \cmidrule(lr){7-8}
        \cmidrule(lr){9-10}
        \cmidrule(lr){11-12}
        \cmidrule(lr){13-14}
        \multicolumn{2}{c}{} &
        \multicolumn{2}{l}{} &
        \multicolumn{2}{c}{$\prescript{2}{2}{U}_\mathtt{Add}$} & 
        \multicolumn{2}{c}{$\prescript{2}{5}{U}_\mathtt{Add}$} & 
        \multicolumn{2}{c}{$\prescript{3}{3}{U}_\mathtt{Add}$}& 
        \multicolumn{2}{c}{$\prescript{3}{5}{U}_\mathtt{Add}$}& 
        \multicolumn{2}{c}{$\prescript{2}{10}{U}_\mathtt{Unit}$} \\
        \midrule
        \multicolumn{2}{c}{Ours} & \multicolumn{2}{l}{\name} & \multicolumn{2}{l}{\textbf{0.878}} & \multicolumn{2}{l}{\textbf{2.31}} & \multicolumn{2}{l}{\textbf{1.6748}} & \multicolumn{2}{l}{\textbf{3.1237}} & \multicolumn{2}{l}{\textbf{1.4294}}\\
        \midrule
        \multicolumn{2}{c}{\multirow{3}{*}{SP baselines}} & \multicolumn{2}{l}{VCG} & \multicolumn{2}{l}{0.667} & \multicolumn{2}{l}{1.667} & \multicolumn{2}{l}{1.4990} & \multicolumn{2}{l}{2.5000} & \multicolumn{2}{l}{$--$}\\
        \multicolumn{2}{c}{} & \multicolumn{2}{l}{Item-Myerson} & \multicolumn{2}{l}{0.833} & \multicolumn{2}{l}{2.083} & \multicolumn{2}{l}{1.5919} & \multicolumn{2}{l}{2.6574} & \multicolumn{2}{l}{$--$}\\
        \multicolumn{2}{c}{} & \multicolumn{2}{l}{AMenuNet} & \multicolumn{2}{l}{0.8628} & \multicolumn{2}{l}{2.2768} & \multicolumn{2}{l}{1.6322} & \multicolumn{2}{l}{2.8005} & \multicolumn{2}{l}{1.2908}\\
        \cmidrule(lr){1-2}
        \cmidrule(lr){3-4}
        \cmidrule(lr){5-6}
        \cmidrule(lr){7-8}
        \cmidrule(lr){9-10}
        \cmidrule(lr){11-12}
        \cmidrule(lr){13-14}
        \multicolumn{2}{c}{\multirow{2}{*}{\makecell{Baselines with\\ IC violation}}} & \multicolumn{2}{l}{RegretNet\tnote{1}} & \multicolumn{2}{l}{0.908} & \multicolumn{2}{l}{2.437} & \multicolumn{2}{l}{1.68057} & \multicolumn{2}{l}{2.65086} & \multicolumn{2}{l}{1.4323}\\
        \multicolumn{2}{c}{} & \multicolumn{2}{l}{IC Violation} & \multicolumn{2}{l}{0.00054} & \multicolumn{2}{l}{0.00146} & \multicolumn{2}{l}{0.00182} & \multicolumn{2}{l}{0.01095} & \multicolumn{2}{l}{0.00487}\\
        \toprule
    \end{tabular}
    \begin{tablenotes}
        \item[1] In all tables, we report the better performance of RegretNet and RegretFormer. Their IC violation is approximated by using gradient ascent to find a good strategic bid that leads to higher utility. The real regret might be higher than shown.
    \end{tablenotes}
\end{table}

\section{Experimental Results}

We conduct a comprehensive set of experiments to evaluate \name~against various baselines, and to understand how \name~works by visualizing and analyzing the learned mechanisms. Specifically, we organize our experiments by answering the following questions: (1) Can \name~outperform existing deep auction methods? How closely do existing solutions such
as those coming from AMAs and RegretNet approach the optimal designs? (Sec.~\ref{sec:exp_emp}); (2) Can \name~recover the optimal auctions for settings where an analytical solution is known? (Sec.~\ref{sec:exp_yao}); (3) How are the mechanisms learned by \name~different from those learned by other deep methods? (Sec.~\ref{sec:exp_didactic} and~\ref{sec:exp_viz});
and (4) How can we construct adaptive grids? (Appendix.~\ref{sec:exp_adaptive_grid}.) 

\begin{figure}
    \vspace{-1.5em}
    \centering
    \includegraphics[width=0.95\linewidth]{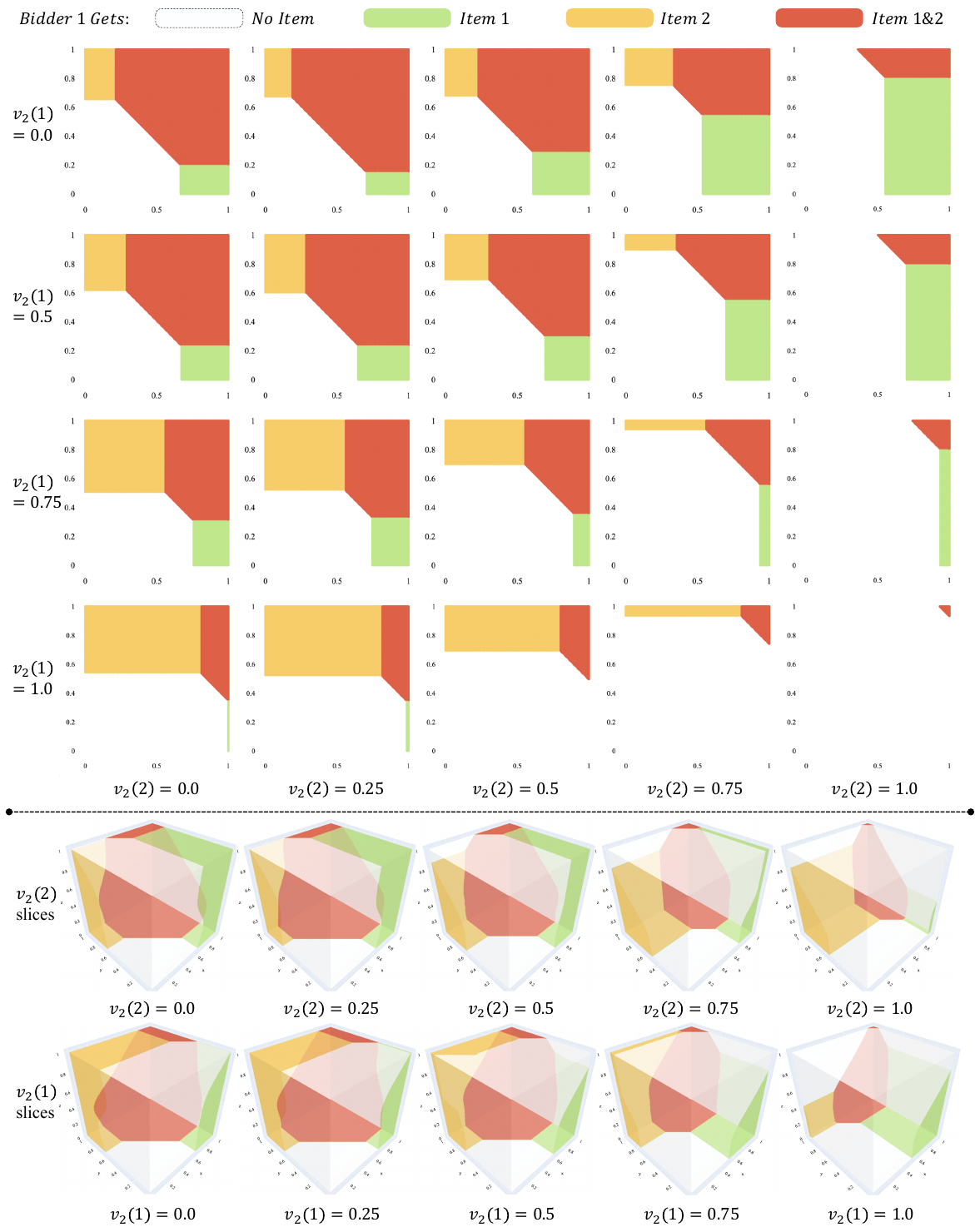}
    \vspace{-1em}
    \caption{Slicing the \name~mechanism in the auction setting with 2 additive bidders, 2 items, and uniform values on $[0,1]$. \textbf{2D Plots}: We vary Bidder 1's values, $v_1(1)$ and $v_1(2)$ in each subplot
    on the x- and y-axis respectively, varying Bidder 2's values, $v_2(1)$ and $v_2(2)$, across each subplot. \textbf{3D Plots}: We first  vary $v_1(1)$, $v_1(2)$, $v_2(1)$ in each subplot on x-, y-, z-axis, respectively, varying  $v_2(2)$ across each subplot.
     We then vary $v_1(1)$, $v_1(2)$, $v_2(2)$ in each subplot, on x-, y-, z-axis, respectively, 
     varying $v_2(1)$ across each subplot.
    \label{fig:u22_slicing}}
    \vspace{-1.5em}
\end{figure}

\subsection{Setup}\label{sec:exp_setup}
\textbf{Baselines}.\ \  We compare our results against the following baselines. (1) The \textbf{VCG} mechanism~\citep{vickrey1961counterspeculation,clarke1971multipart,groves}, which is SP. (2) \textbf{Item-Myerson}. This is SP for additive valuations, and independently runs an optimal Myerson auction for each individual item. (3) \textbf{Affine Maximizers}. We compare with {\em AMenuNet}~\citep{duan2023scalable}, which  is the state of the art and uses the transformer architecture to learn an affine transformation. Given that AMenuNet has been compared to other affine maximization methods, such as {\em AMA Lottery}~\citep{curry2022learning}, we exclude these works from our direct comparison. (4) \textbf{RegretNet}~\citep{duetting2023optimal} and \textbf{RegretFormer}~\citep{ivanov2022optimal}. This serves as a prominent example of methods are that expressive but not exactly SP.

\textbf{Network architecture and training setup}. For training networks, we use a minibatch of size $2^{13}$ and a menu size $K=300$, unless the number of items is greater than 5, in which case we increase the menu size to 1000 to strengthen the representation capacity. During training, we gradually increase the SoftMax temperature $\lambda_{\textsc{SoftMax}}$ in Eq.~\ref{equ:softmax_in_loss}, as smaller values of $\lambda_{\textsc{SoftMax}}$ help with initial exploration of the network weights and larger values approximate the argmax operation better. We also gradually increase the scaling factor $\lambda_{\textsc{Incomp}}$.   We evaluate the performance of the network once every $200$ epochs using a test set of $200K$ samples. Of all the checkpoints satisfying  an upper bound of compatibility violation ($0.1\%$ or $0.5\%$,  depending on task), we pick the one with the best revenue for   menu transformation.  Appendix~\ref{appx:exp_gem} 
gives full details regarding training and price adjustment.

\subsection{Representation Capacity}
\label{sec:exp_didactic}

We first demonstrate the representation capacity of \name\ by considering a setting 
 with two additive bidders, two items,
and i.i.d.~uniform values on $[0,1]$ ($\prescript{2}{2}{U}_\mathtt{Add}$).
This also serves to showcase \name\ as a method that  provides a 
first look at the structure 
of what we conjecture the optimal auction design for this problem.
%
%
In Fig.~\ref{fig:mech_compare_2x2}, we fix Bidder 2's value for the items to $(0, 0.6)$, 
and show the menu element selected for Bidder 1 with different values. 
We compare the mechanism learned by \name~(pre- and post-transformation), RegretNet, and AMenuNet. 
For fairness, AMenuNet's menu size was increased to 2048 from the default 512, enhancing its capacity. 
The comparison with AMenuNet highlights the loss in expressive  capacity of 
 affine-maximizing mechanisms. In particular, AMenuNet is unable to learn the top-left region where bidder 1 gets item 2 but not item 1 and instead has a larger set of types for which bidder 1 receives no items and makes zero payment.
 This serves to illustrate that \name\  is capable of achieving higher revenue than AMenuNet.
The comparison to  RegretNet  is also  interesting.  Compared with AMenuNet, RegretNet does an arguably better
job of identifying the high-level structure of the optimal allocation and payment rule. However, it is 
``fuzzy" and rendered with a lack of crispness as to the boundaries of regions. In contrast, \name\
extends the advantage of menu-based methods to 
this multi-bidder setting, with the bidder-optimizing aspect
providing a clear decision boundary between regions. 
 This showcases  the improved interpretability  of \name\ over 
 RegretNet.
 Also, Table~\ref{tab:u01_baselines} shows the improved revenue of \name\ in this setting over the
SP baselines.
\begin{table} [t]
    \caption{$n$ additive bidders, two items, and valuations on  support size two~\citep{yao2017dominant}. The optimality of \name~is consistent across different settings: first for $a=3$, $n=2$, $p=0.3$ and  varying $b\in\{7, 5, 4, 7/2\}$; and second with  $a=3$, $b=4$, $p=0.3$ and varying $n\in\{3,5,8\}$.
It is interesting that the transformation can {\em increase} the revenue and
 realize the exact optimal revenue. See Fig.~\ref{fig:mech_compare_yao_2x2} for a 
 visualization of the learned auction.\label{tab:yao}}
    \centering
    \begin{tabular}{CRCRCRCRCRCR|CRCRCRCR}
        \toprule
        \multicolumn{2}{c}{\multirow{2}{*}{}} &
        \multicolumn{2}{l}{\multirow{2}{*}{Alg.}} &
        \multicolumn{14}{c}{Setting ($a$ fixed to 3; $p$ fixed to 0.3.)}\\
        
        \cmidrule(lr){5-6}
        \cmidrule(lr){7-8}
        \cmidrule(lr){9-10}
        \cmidrule(lr){11-12}
        \cmidrule(lr){13-14}
        \cmidrule(lr){15-16}
        \cmidrule(lr){17-18}
        \multicolumn{2}{c}{} &
        \multicolumn{2}{l}{} &
        \multicolumn{2}{c}{$b$=7} & 
        \multicolumn{2}{c}{$b$=5} & 
        \multicolumn{2}{c}{$b$=4} & 
        \multicolumn{2}{c}{$b$=$\frac{7}{2}$} & 
        \multicolumn{2}{c}{$n$=3} & 
        \multicolumn{2}{c}{$n$=5} & 
        \multicolumn{2}{c}{$n$=8}
        \\
        \midrule
        
        \multicolumn{2}{c}{\multirow{2}{*}{Ours}} & \multicolumn{2}{l}{Pre-Transf.} & \multicolumn{2}{l}{12.7077} & \multicolumn{2}{l}{9.1435} & \multicolumn{2}{l}{7.4735} & \multicolumn{2}{l}{6.7186} & \multicolumn{2}{l}{7.8220} & \multicolumn{2}{l}{7.8220} & \multicolumn{2}{l}{7.9954}\\
        \cmidrule(lr){3-4}
        \cmidrule(lr){5-6}
        \cmidrule(lr){7-8}
        \cmidrule(lr){9-10}
        \cmidrule(lr){11-12}
        \cmidrule(lr){13-14}
        \cmidrule(lr){15-16}
        \cmidrule(lr){17-18}
        \multicolumn{2}{c}{} & \multicolumn{2}{l}{Post-Transf.} & \multicolumn{2}{l}{\textbf{12.7400}} & \multicolumn{2}{l}{\textbf{9.1504}} & \multicolumn{2}{l}{\textbf{7.4774}} & \multicolumn{2}{l}{\textbf{6.72205}} & \multicolumn{2}{l}{\textbf{7.8309}} & \multicolumn{2}{l}{\textbf{7.9830}} & \multicolumn{2}{l}{\textbf{7.9996}}\\
        \midrule
        \multicolumn{2}{c}{\multirow{5}{*}{Baselines}} & \multicolumn{2}{l}{VCG} & \multicolumn{2}{l}{9.9221} & \multicolumn{2}{l}{7.9668} & \multicolumn{2}{l}{6.9815} & \multicolumn{2}{l}{6.4908} & \multicolumn{2}{l}{7.5673} & \multicolumn{2}{l}{7.9363} & \multicolumn{2}{l}{7.9977}\\
        \multicolumn{2}{c}{} & \multicolumn{2}{l}{Item-Myerson} & \multicolumn{2}{l}{10.9220} & \multicolumn{2}{l}{7.2829} & \multicolumn{2}{l}{5.4607} & \multicolumn{2}{l}{4.5490} & \multicolumn{2}{l}{5.8364} & \multicolumn{2}{l}{5.9848} & \multicolumn{2}{l}{5.9996}\\
        \multicolumn{2}{c}{} & \multicolumn{2}{l}{AMenuNet} & \multicolumn{2}{l}{10.2852} & \multicolumn{2}{l}{8.1416} & \multicolumn{2}{l}{7.0697} & \multicolumn{2}{l}{6.5349} & \multicolumn{2}{l}{7.5323} & \multicolumn{2}{l}{7.7363} & \multicolumn{2}{l}{7.9956}\\
        \cmidrule(lr){3-4}
        \cmidrule(lr){5-6}
        \cmidrule(lr){7-8}
        \cmidrule(lr){9-10}
        \cmidrule(lr){11-12}
        \cmidrule(lr){13-14}
        \cmidrule(lr){15-16}
        \cmidrule(lr){17-18}
        \multicolumn{2}{c}{\multirow{2}{*}{}} & \multicolumn{2}{l}{RegretNet} & \multicolumn{2}{l}{12.8052} & \multicolumn{2}{l}{9.17377} & \multicolumn{2}{l}{7.50169} & \multicolumn{2}{l}{6.6799} & \multicolumn{2}{l}{7.8370} & \multicolumn{2}{l}{7.998} & \multicolumn{2}{l}{8.0000} \\
        \multicolumn{2}{c}{} & \multicolumn{2}{l}{(IC violation)} & \multicolumn{2}{l}{0.00951} & \multicolumn{2}{l}{0.00846} & \multicolumn{2}{l}{0.00548} & \multicolumn{2}{l}{0.0094} & \multicolumn{2}{l}{0.0383} & \multicolumn{2}{l}{0.0396} & \multicolumn{2}{l}{0.0363} \\
        \toprule
        \multicolumn{2}{c}{Optimal} & \multicolumn{2}{l}{\citet{yao2017dominant}} & \multicolumn{2}{l}{\textbf{12.7400}} & \multicolumn{2}{l}{\textbf{9.1504}} & \multicolumn{2}{l}{\textbf{7.4774}} & \multicolumn{2}{l}{\textbf{6.72205}}  & \multicolumn{2}{l}{\textbf{7.8309}} & \multicolumn{2}{l}{\textbf{7.9840}} & \multicolumn{2}{l}{\textbf{7.9996}}\\
        \toprule
    \end{tabular}
\end{table}

\subsection{Benchmarking \name~Performance}\label{sec:exp_emp}

We benchmark \name~on auction settings with more than one bidder that have been studied in the previous literature~\citep{duetting2023optimal,curry2020certifying,duan2023scalable}.

In Table~\ref{tab:u01_baselines}, we show results on settings where bidders' values come from the uniform distribution in the range $[0,1]$. The details of the setting are represented by $\prescript{n}{m}{U}_\mathtt{Add/Unit}$, where $m$ is the number of items, $n$ is the number of bidders, and $\mathtt{Add/Unit}$ indicates the bidders' valuation type (additive or unit-demand). \name~consistently outperforms all SP baselines  and by a large margin. Moreover, the revenue of \name\ is very close to that achieved by RegretNet; e.g.,  0.8979 in $\prescript{2}{2}{U}_\mathtt{Add}$, while RegretNet gets 0.908 (as a reminder, RegretNet is not fully SP).
 These results respond to the questions raised in the introduction. They confirm  a gap between the optimal auction and the solutions derived by AMA approaches:  even if \name\ is not optimal, it is SP, and thus a valid lower bound on the revenue from the exactly optimal design.
Moreover, under our working hypothesis that \name\ finds designs that are   close to
 optimal, these results also show that the revenue of RegretNet  is  not too much higher than the optimal revenue.
 In Table~\ref{tab:be_baselines}, we also show how \name~generalizes to other distributions, 
 considering
 the $Beta$ distribution with parameters $\alpha = 1, \beta = 2$, as well as an irregular distribution, where each bidder's valuation is 
$U[0, 3]$ with probability $3/4$ and  $U[3, 8]$ with probability $1/4$~\citep{hartline2013mechanism}.
These cases serve to validate the strong performance of \name, For example, Myerson is optimal in the $n=1$ case with irregular valuation (making use of ironing), and \name\ achieves almost the same  revenue.
%

\subsection{Recovering Theoretically Optimal, Multi-Bidder Designs}
\label{sec:exp_yao}

To our knowledge, the only theoretical work that handles settings involving two or more bidders
and items is the study by~\citet{yao2017dominant}. 
The setting considered has   two items, $n>1$ additive bidders, and valuations sampled i.i.d.~from a distribution $Pr\{v_i(j)=a\}=p$ and $Pr\{v_i(j)=b\}=1-p$, for different values of $a$ and $b$. \citet{yao2017dominant} proves the maximum revenue achievable under any SP and IR auction is
$\mathsf{R}_{n,a,b,p} = 2(1-p^n)b + p_0\left[2a-\frac{1-p^2}{p^2}(b-a)\right]_+ + p_1\left[a-\frac{1-p}{2p}(b-a)\right]_+ + p_2\left[a-\frac{1-p}{p}(b-a)\right]_+$,
where $[x]_+=max\{x,0\}$, $p_0=p^{2n}$, $p_1=2np^{2n-1}(1-p)$, and $p_2=2p^{n}(1-p^n-np^{n-1}(1-p))$.
The optimal auction  depends on which of the following intervals contains $b$: 
%
%
$b\in\left(a,\frac{1+p^2}{1+p^2}a\right)$; $\left[\frac{1+p^2}{1+p^2}a, \frac{1}{1-p}a\right)$; $\left[\frac{1}{1-p}a,\frac{1+p}{1-p}a\right)$; $\left[\frac{1+p}{1-p}a,\infty\right).$

We evaluate \name~in all four cases, and with different numbers of bidders. 
We fix $p=0.3$, $a=3$ and select $b$ from $\{3.5, 4, 5, 7\}$. For $a=3$ and $b=4$, we vary $n$ within $\{2,3,5,8\}$. 
In Table~\ref{tab:yao}, we show that \name~accurately recovers the optimal revenue in each setting. By contrast, 
the other deep methods, except for  RegretNet, which is not quite SP,
 do not approximate optimality.
 Fig.~\ref{fig:mech_compare_yao_2x2} further validates that \name's mechanism aligns with the optimal, echoing the auction prescription of Algorithm 1 from~\citet{yao2017dominant}.
Moreover, price adjustments align prices with the optimal solution, notably in Fig.~\ref{fig:mech_compare_yao_2x2}'s first two columns. Although the deep learning phase secures optimal allocation, it requires post-processing for accurate pricing. This highlights the role 
of price adjustment 
beyond fixing menu incompatibility, significantly contributing
here to overall revenue optimality. It is also interesting that AMenuNet comes closer than RegretNet  to the optimal in this setting.\footnote{Since the valuation domain in this example is discrete, a linear program can also solve Yao's problem. 
It has  $O(2^{2n})$ variables and constraints, where the IC constraints enumerate each agent's possible deviations.
For the experiments  in Table~\ref{tab:yao}, with moderate sizes of $n$, this linear program can be efficiently solved.
Our main focus in this paper  is on auctions with continuous valuation, of which there's no known, SP
linear program formulation.  
 However, since we know of no  theoretical results  for multi-buyer multi-item auctions except for~\citet{yao2017dominant}, we include this as an additional validation that \name~is able to learn optimal auctions,
 lending credibility to its performance.}

\begin{figure*}
    \centering
    \includegraphics[width=\linewidth]{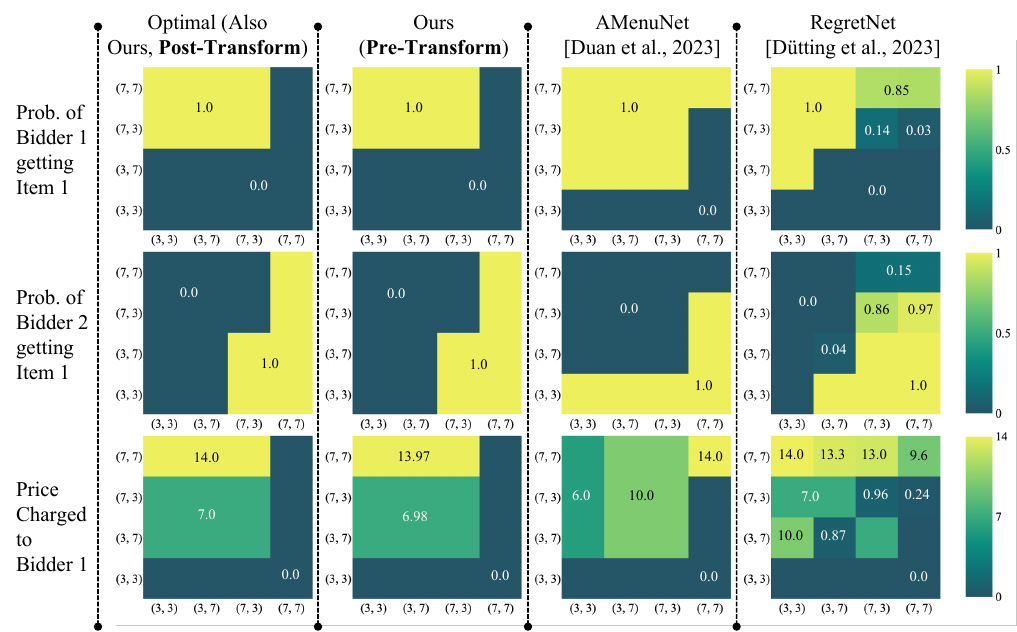}
    \vspace{-2em}
    \caption{2 additive bidders, 2 items, and valuations on support size two~\citep{yao2017dominant} with $n=2$, $a=3$, $b=7$, $p=0.3$. In each subplot, x- and y-axis is the values of Bidder 1 and 2, respectively, on each of the two items. \name~recovers the optimal auction after transformation (Column 1), while other deep methods cannot (Columns 3, 4). The deep learning phase of \name~provides the perfect allocation mechanism (Column 2), but it does not accurately set  prices. The menu transformation resolves this discrepancy (Column 1), with  this price adjustment having a critical role in realizing revenue optimality. Fig.~\ref{fig:mech_compare_yao_2x2_extended} shows the full mechanism.
    \label{fig:mech_compare_yao_2x2}}
\end{figure*}

\subsection{Mechanism Analysis}
\label{sec:exp_viz}

In this section, we seek to better understand the improved performance achieved by \name~by 
visualizing the learned auction rules.

(I) $\prescript{2}{2}{U}_\mathtt{Add}$. In Fig.~\ref{fig:mech_compare_2x2}, we  already
compared the mechanisms learned in this case 
by different methods. 
Here, we take a closer look at the  auction rule learned by \name, using
both 2-dimensional (2D) and 3-dimensional (3D) visualizations 
to gain an 
 understanding of its behavior under different
 bidder valuations.

\textbf{2D Plots Analysis.}\ \ 
In the  2D visualizations of Fig.~\ref{fig:u22_slicing}, 
we vary the values of Bidder 2, specifically $v_2(1)$ and $v_2(2)$, across different subplots. Each subplot plots the allocation to Bidder
1 as its value for item 1 ($v_1(1)$) and item 2 ($v_1(2)$) 
vary along the x- and y-axes, respectively. 
%
The learned auction conforms to a specific structure which is similar to the 
optimal single-bidder structure~\citep{manelli2006bundling}, revealed
here for the first time through \name. \textbf{3D Plots Analysis.}
Moving to the 3D plots in  Fig.~\ref{fig:u22_slicing}, 
we fix the value of item 2 for Bidder 2 ($v_2(2)$)
to different values 
and examine the interaction between $v_1(1)$, $v_1(2)$, and $v_2(1)$ across the x-, y-, and z-axes, respectively. 
We further fix the value of $v_2(1)$ to different values,
and vary $v_1(1)$, $v_1(2)$, and $v_2(2)$ in the second row of the 3D plot.
It is interesting to see that 
these plots are characterized by a clear symmetric structure.

(II) $\prescript{3}{3}{U}_\mathtt{Add}$. In Fig.~\ref{fig:u33_compare}, we explore the learned \name\ auction design for a setting with three additive bidders competing for three items, with their values uniformly distributed  on $[0,1]$.
The analysis   fixes the valuations of  Bidders 2 and 3 at 
$v_2=v_3=(0.2,0.2,0.2)$, and looks
to understand how Bidder $1$'s value for each item,
 represented across three axes,
 influences its own allocation.
It is  interesting to compare with the 
optimal allocation structure for a simpler, single-bidder  three-item setting with the same value distribution, as previously solved by~\citet{giannakopoulos2014duality} (their Straight-Jacket Auction). 
%
We also compare AMenuNet, RegretNet, and  \name~in Fig.~\ref{fig:u33_compare} (b), (c), and (d), respectively. 
AMenuNet  learns a very sub-optimal allocation rule in this setting,
 even allocating an item to Bidder 1  when its value for the item is zero; e.g., for $v_1=(0.6,\ 0,\ 0.3)$.
 This contrasts sharply with the allocation rule obtained by \name, which has a structure that closely mirrors the optimal,  one-bidder structure.
 Correspondingly, as shown in Table~\ref{tab:u01_baselines}, AMenuNet achieves lower revenue (1.6322) than \name~(1.6546).
 RegretNet is closer at capturing the macro-structure in the \name\ design
 but loses the specific detail revealed by \name.
Further tests with \name~examine how the 
allocation to Bidder 1 changes for different  valuations of  Bidders 2 and 3. We  observe that \name~continues to maintain a similar high-level structure in the allocation rule. However, the structure might exhibit asymmetric, as illustrated in Fig.~\ref{fig:u33_compare} (e, f).
\begin{table} [t]
    \caption{Bidder values coming from $Beta$ ($\alpha=1,\beta=2$) and from an irregular distribution ($U[0,3]$ with probability $3/4$ and $U[3,8]$ with probability $1/4$). \citet{myerson1981optimal} with ironing gives the optimal solution in single-item auctions with the irregular distribution, where \name~also achieves near-optimal revenue. In all other cases, where the optimal solution is unknown, our method outperforms all SP baselines.\label{tab:be_baselines} }
    \centering
    \begin{tabular}{CRCRCRCRCRCRCR}
        \toprule
        \multicolumn{2}{c}{\multirow{2}{*}{}} &
        \multicolumn{2}{l}{\multirow{2}{*}{Alg.}} &
        \multicolumn{10}{c}{Setting}\\
        
        \cmidrule(lr){5-6}
        \cmidrule(lr){7-8}
        \cmidrule(lr){9-10}
        \cmidrule(lr){11-12}
        \cmidrule(lr){13-14}
        \multicolumn{2}{c}{} &
        \multicolumn{2}{l}{} &
        \multicolumn{2}{c}{$\prescript{2}{2}{Beta}_\mathtt{Add}$} & 
        \multicolumn{2}{c}{$\prescript{2}{5}{Beta}_\mathtt{Add}$} &
        \multicolumn{2}{c}{$\prescript{3}{2}{Beta}_\mathtt{Add}$} & 
        \multicolumn{2}{c}{$\prescript{3}{1}{IRR}_\mathtt{Add}$}& 
        \multicolumn{2}{c}{$\prescript{3}{2}{IRR}_\mathtt{Add}$} \\
        \midrule
        
        \multicolumn{2}{c}{Ours} & \multicolumn{2}{l}{\name} & \multicolumn{2}{l}{\textbf{0.5575}} & 
        \multicolumn{2}{l}{\textbf{1.5598}} & \multicolumn{2}{l}{\textbf{0.7333}} & \multicolumn{2}{l}{2.3462} & \multicolumn{2}{l}{\textbf{5.0197}}\\
        \midrule
        \multicolumn{2}{c}{\multirow{3}{*}{SP baselines}} & \multicolumn{2}{l}{VCG} & \multicolumn{2}{l}{0.4008} & 
        \multicolumn{2}{l}{0.9992} & \multicolumn{2}{l}{0.6298} & \multicolumn{2}{l}{2.2131} & \multicolumn{2}{l}{4.4423}\\
        \multicolumn{2}{c}{} & \multicolumn{2}{l}{Item-Myerson} & \multicolumn{2}{l}{0.5144} & \multicolumn{2}{l}{1.2814} &
        \multicolumn{2}{l}{0.6775} & \multicolumn{2}{l}{2.3680} & \multicolumn{2}{l}{4.7360}\\
        \multicolumn{2}{c}{} & \multicolumn{2}{l}{AMenuNet} & \multicolumn{2}{l}{0.5494} & \multicolumn{2}{l}{1.5098} & 
        \multicolumn{2}{l}{0.6897} & \multicolumn{2}{l}{2.3369} & \multicolumn{2}{l}{4.9277}\\
        \cmidrule(lr){1-2}
        \cmidrule(lr){3-4}
        \cmidrule(lr){5-6}
        \cmidrule(lr){7-8}
        \cmidrule(lr){9-10}
        \cmidrule(lr){11-12}
        \cmidrule(lr){13-14}
        \multicolumn{2}{c}{\multirow{2}{*}{\makecell{Baselines with\\ IC violation}}} & \multicolumn{2}{l}{RegretNet} & \multicolumn{2}{l}{0.5782} & \multicolumn{2}{l}{1.51340} & 
        \multicolumn{2}{l}{0.7337} & \multicolumn{2}{l}{2.5307} & \multicolumn{2}{l}{5.30577}\\
        \multicolumn{2}{c}{} & \multicolumn{2}{l}{IC Violation} & \multicolumn{2}{l}{0.0042} & \multicolumn{2}{l}{0.01060} & 
        \multicolumn{2}{l}{0.00324} & \multicolumn{2}{l}{0.0078} & \multicolumn{2}{l}{0.01083}\\
        \toprule
        \multicolumn{2}{c}{Optimal} & \multicolumn{2}{l}{Myerson} & \multicolumn{2}{l}{$--$} & \multicolumn{2}{l}{$--$} & \multicolumn{2}{l}{$--$} & \multicolumn{2}{l}{\textbf{2.3680}} & \multicolumn{2}{l}{$--$}\\
        \toprule
    \end{tabular}
\end{table}

\section{Closing Remarks}

In this paper, we have introduced the first expressive (i.e., general), strategy-proof, method for learning revenue-optimizing,
multi-bidder and multi-item auctions under additive and unit-demand valuations.
The innovation of \name\ is to  use
a menu-based representation for the learned, multi-bidder
 auction designs, which provides both exact strategy-proofness as well as interpretability.  
 The technical challenge is to achieve menu compatibility, so that the bidder-optimizing choice for each bidder can be selected simultaneously  without leading to infeasibility (i.e., no item is allocated more than once). We achieve this through a specific choice of incompatibility loss during training and by carefully transforming prices in
learned menus, post-training, to achieve  compatibility margins that are sufficient to ensure compatibility with probability 1 throughout the value domain given Lipschitz smoothness of the learned networks. \name~outperforms all previous SP baselines, and reveals for the first time the structure in what we conjecture
to be essentially optimal auction design for various multi-bidder multi-item settings considered in this paper. There are many directions for future work, including establishing sample complexity bounds similar to~\citet{duetting2023optimal}, further scaling-up the menu transformation method, developing open-access tools for use by theoretical economists, exploring whether RegretNet provides a useful (i.e., empirical) upper-bound on revenue coming from \name, extending duality theory to prove optimality in these multi-bidder settings~\citep{cai2016duality}, extending from additive and unit-demand valuations to problems with combinatorial valuations, improved visualization methods, adopting other objectives, and considering other mechanism design problems.

\newpage
\bibliographystyle{ACM-Reference-Format}
\bibliography{sample-bibliography}

\newpage
\appendix

\section{Experiments}

\subsection{The value of Big \texorpdfstring{$M$}{M}}\label{appx:big_M}

In this section, we discuss the appropriate values for $M$. Recall that the constraint involving $M$ in the MILP is:
\begin{align}
    U_{(\ell)} &\le \vv_{i,(\ell)}^\Tau \bm{\alpha}_i^{(k)} - \beta_i^{(k)} - \Delta\beta_i^{(k)} + (1 - z^{(\ell k)}) M.\label{eq:10b}
\end{align}
For a fixed $\ell$, we want $z^{(\ell k^*)}=1$ for the optimal menu item $k^*$ that gives the max utility to bidder $i$ with value $\vv_{i,(\ell)}$, and $z^{(\ell k)}=0$ for all the other menu items $k\neq k^*$. When $z^{(\ell k^*)}=1$, and also considering constraint, $U_{(\ell)} \ge \vv_{i,(\ell)}^\Tau\bm\alpha_i^{(k)} - \beta_i^{(k)} - \Delta\beta^{(k)}_i + (1-z^{(\ell k)})s_m$,
 we have 
\begin{align}
U_{(\ell)} = \vv_{i,(\ell)}^\Tau \bm{\alpha}_i^{(k^*)} - \beta_i^{(k^*)} - \Delta\beta_i^{(k^*)}.\label{eq:10c}
\end{align}
We expand Constraint \ref{eq:10b} using Eq. \ref{eq:10c}:
\begin{align}
    \vv_{i,(\ell)}^\Tau \bm{\alpha}_i^{(k^*)} - \beta_i^{(k^*)} - \Delta\beta_i^{(k^*)} \le \vv_{i,(\ell)}^\Tau \bm{\alpha}_i^{(k')} - \beta_i^{(k')} - \Delta\beta_i^{(k')} + M,\ \ \ \forall k'\ne k^*.\nonumber
\end{align}
It follows that for feasibility we need,
\begin{align}
    M\ge \vv_{i,(\ell)}^\Tau \bm{\alpha}_i^{(k^*)} - \beta_i^{(k^*)} - \Delta\beta_i^{(k^*)} - (\vv_{i,(\ell)}^\Tau \bm{\alpha}_i^{(k')} - \beta_i^{(k')} - \Delta\beta_i^{(k')}),\ \ \ \forall k'\ne k^*,\nonumber
\end{align}
and $M$ should be at least the gap between the maximum and minimum possible utility of all compatible menu elements. We first focus on the additive valuation function, and derive an upper and lower bound on bidder utility in the following lemma.
%
\begin{lemma}[Menu Element Utility Bounds]\label{lemma:utility_bound}
The following are valid utility bounds on menu elements, for an additive valuation bidder, $m\geq 1$ items, and $v_{\max}>0$ maximum value for a single item,
  \begin{align}
      \min_{i, k, \ell} \left[\vv_{i,(\ell)}^\Tau \bm{\alpha}_i^{(k)} - \beta^{(k)}_i - \Delta\beta^{(k)}_i\right] \ge -m\cdot v_{\max},\\
      \max_{i, k, \ell} \left[\vv_{i,(\ell)}^\Tau \bm{\alpha}_i^{(k)} - \beta^{(k)}_i - \Delta\beta^{(k)}_i\right] \le 2m\cdot v_{\max}.
  \end{align}
\end{lemma}
\begin{proof}
    For the lower bound, suppose that for some $i, \ell, k$, $\vv_{i,(\ell)}^\Tau \bm{\alpha}_i^{(k)} - \beta_i^{(k)} - \Delta\beta_i^{(k)} \ge -m\cdot v_{\max}$ does not hold. Since values are non-negative, we must have $\beta_i^{(k)} + \Delta\beta_i^{(k)} > m \cdot v_{\max}$ for the price. However, when a menu element has a price greater than $m \cdot v_{\max}$, it can never be selected. This is because the value of items a bidder can get is upper bounded by $\vv_{i,(\ell)}^\Tau \bm{\alpha}_i^{(k)} \leq m \cdot v_{max}$. If a price is higher than $m \cdot v_{\max}$, the corresponding menu element will always be less attractive than the IR option (no allocation, and zero price). Therefore, we can construct an equivalent mechanism by removing all menu elements with a price larger than $m \cdot v_{\max}$. 
    
    
    For the upper bound, suppose that for some $i, \ell, k$, $\vv_{i,(\ell)}^\Tau \bm{\alpha}_i^{(k)} - \beta_i^{(k)} - \Delta\beta_i^{(k)} \le 2m \cdot v_{\max}$ does not hold. Since the value of items a bidder can get is upper bounded by $\vv_{i,(\ell)}^\Tau \bm{\alpha}_i^{(k)} \leq m \cdot v_{max}$, we must have $\beta_i^{(k)} + \Delta\beta_i^{(k)} < -m \cdot v_{\max}$ for the price. Such a low price means that the bidder, whatever its value is, will always choose items with a negative price. This is because for any menu element with a positive price, its utility to the bidder is less than the values of all items, which is less than $m \cdot v_{\max}$. In this way, the revenue for the seller is negative, and the seller can do better by only giving the null option (0 allocation and 0 price) in the menu without introducing over-allocation. 
\end{proof}

Based on Lemma~\ref{lemma:utility_bound}, $M \geq 3m \cdot v_{\max}$ should be sufficiently large. The analysis for unit-demand valuations is similar, and we can derive that $M \geq 3v_{\max}$ is sufficiently large in this case. These bounds used for deriving the value of $M$ are  loose. Empirically, it is possible to make the big-$M$ values tighter by considering the actual menu.


\subsection{Setup of Our Framework}\label{appx:exp_gem}
We train the neural network \name~on a single NVIDIA Tesla V100 GPU. 

For the price adjustment process, the ``big-$M$" is set to 10. Since the price charged to buyers are non-negative, the utility of a buyer is at most the upper bound of its valuation times the number of items. Since both Uniform and Beta distributions are upper bounded by $1$, setting $M=10$ suffices for the experiments in this paper. We use the Gurobi optimizer~\citep{gurobi} to solve the MILPs, and we adhere to the default parameters provided by Gurobi. The transformation uses a Linux machine with $2$ vcpus and $64$ GB memory. 

\begin{figure*}
    \centering
    \includegraphics[width=0.9\linewidth]{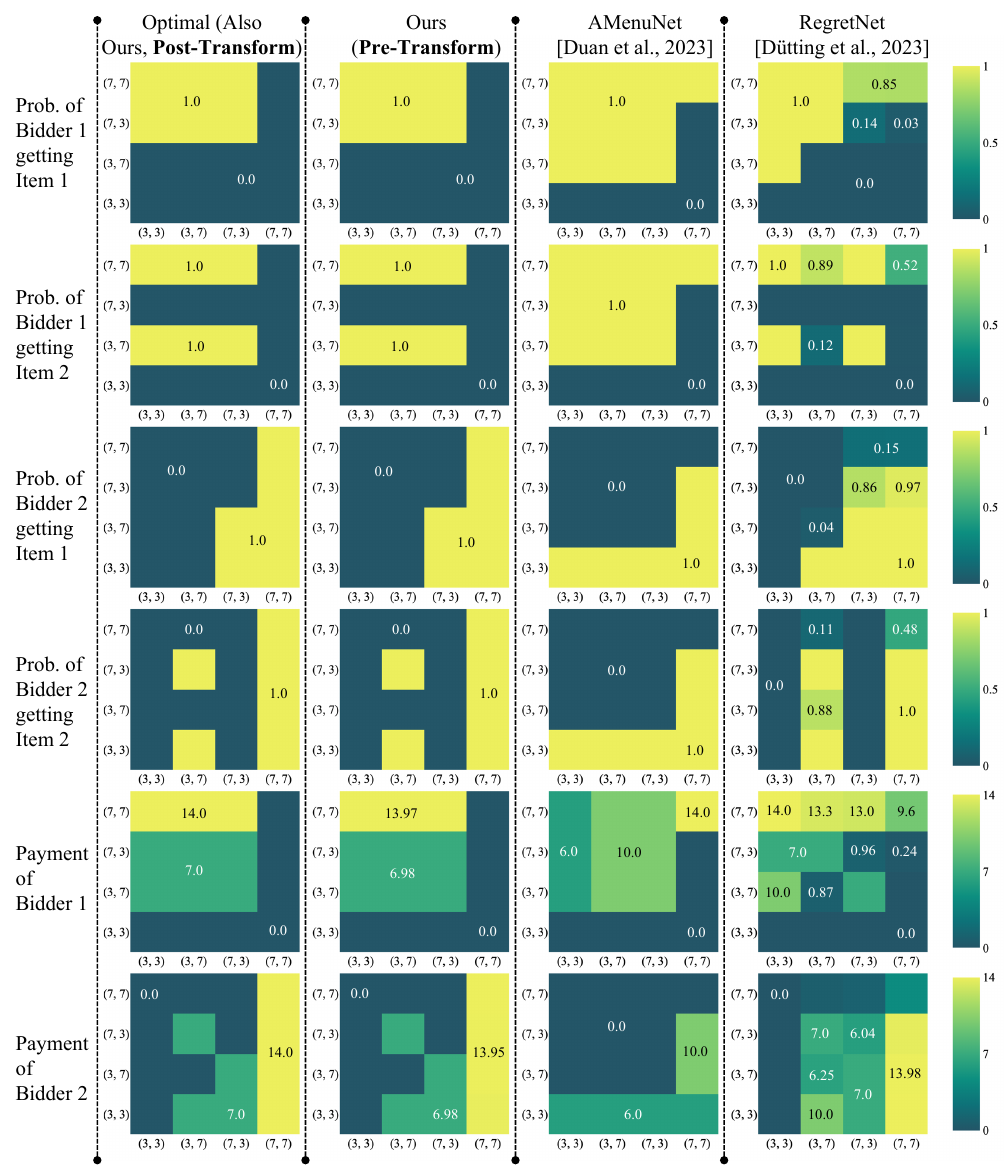}
    \caption{The full mechanisms learned by different methods on the problem with 2 additive bidders, two items, and valuations on support size two~\citep{yao2017dominant} ($n=2$, $a=3$, $b=7$, $p=0.3$). In every subplot, the x-axis represents the bids of the first bidder and the y-axis those of the second bidder for both items. \name~successfully identifies the optimal mechanism post-transformation (Column 1). In contrast, other deep auction learning methods cannot (Columns 3 and 4). The deep learning phase of \name~achieves the ideal allocation mechanism (Column 2); however, it falls short in setting precise prices. This issue is addressed through the menu transformation process (Column 1), demonstrating that the price adjustment plays a pivotal role in achieving optimal revenue.}
    \label{fig:mech_compare_yao_2x2_extended}
\end{figure*}

\subsection{Baselines}\label{appx:exp_baseline}

For Baselines, we use the  GitHub repository in \citet{duan2023scalable} for our AMenuNet baselines results and~\citet{ivanov2022optimal} for the RegretNet and RegretFormer baselines. For AMenuNet, we use the default hyperparameters for the transformer architecture, running a minimum of $2000$ iterations of training steps until it converges. For RegretNet, we use the default setting of 3 hidden layers, each containing $100$ neurons. We run $50$ steps of gradient descent for regret calculation in each iteration, and a total of $800,000$ iterations of training, or less if it converges.
%

\subsection{Adaptive grid}
\label{sec:exp_adaptive_grid}

\begin{figure}
    \centering
    \includegraphics[width=\linewidth]{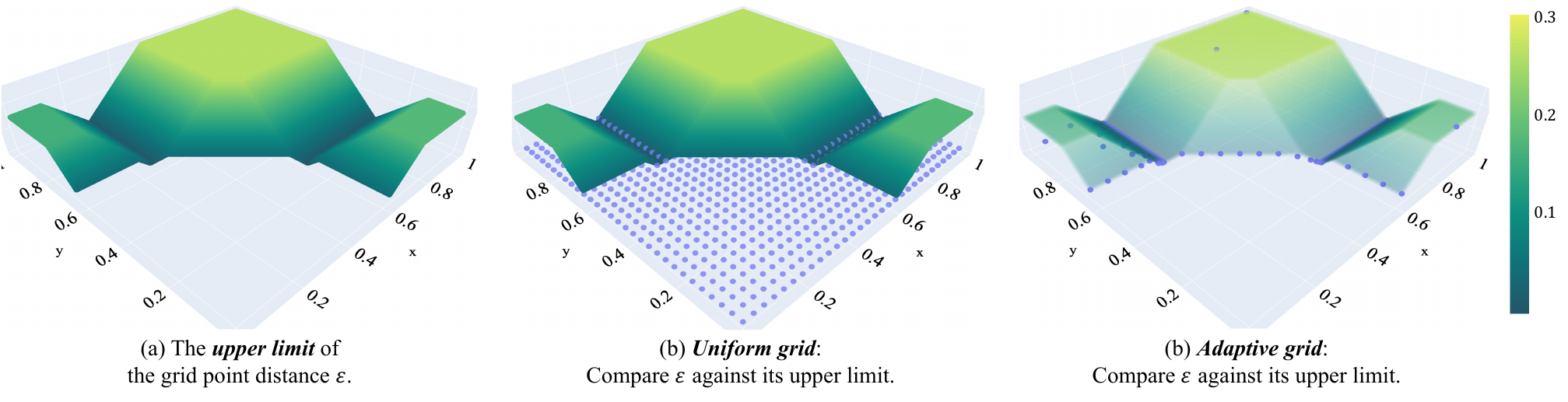}
    \caption{Adaptive grid for the setting with two additive bidders, two items, and i.i.d. values from the uniform distribution on $[0,1]$. Fix $v_2=(0, 0.6)$, and consider the grid of Bidder 1. (a) The upper limit of the distance between two adjacent grid points ($\epsilon$) that can still guarantee menu compatibility. (b) For a uniform, $30\times 30$ grid, $\epsilon=1/30$ is significantly smaller than the upper limit for the majority of the value space, giving us the opportunity to generate an adaptive grid. (c) The adaptive grid: we remove 95\% grid points, and now $\epsilon$ is just slightly below its upper limit.    \label{fig:adaptive_grid}}
\end{figure}

In menu transformation, we use a grid of bidder values to adjust prices
in learned menus. As discussed in Sec.~\ref{sec:method:rdc},
the density (or size) of this grid affects the computational efficiency of this MILP-based adjustment method. In this section, we 
introduce a method to adaptively reduce the density of the grid by exploiting the local Lipschitz smoothness of the learned menu networks.

The major factor that we need to consider when constructing an adaptive grid is to ensure that the proof of our menu compatibility theorem (Theorem~\ref{thm:2_bidder_comp_updated} and~\ref{thm:general_comp}) can still go through. Specifically, the distance in the $\ell_\infty$-norm of two adjacent grid points, $\epsilon$, controls the value of two safety margins $s_f= \epsilon\cdot L_a$ and $s_m=L_a(m \cdot v_{\max}\epsilon+m\epsilon^2)+\epsilon L_p+\epsilon$. The requirement is that the sum bundle of bidders is at least $n\cdot s_f$ smaller than $1$,
and the utility difference between the best and the second best element in a menu is larger than $s_m>0$:
\begin{align}
    & 1 - \sum_i \bm\alpha^*_i \ge n\cdot s_f; \\
    & \forall i, \ u_i^{(k^*)}(\vv_{\shortn i}) - u_i^{(k)}(\vv_{\shortn i}) \ge s_m.\label{eq:eps_upper_limit}
\end{align}
To adaptively adjust the distance between adjacent grid points, we can empirically calculate the actual value of $\tilde{s}_f$ (one minus the actual sum bundle) and $\tilde{s}_m$ (the actual utility difference between the best and the second best element) and thereby calculate the upper bound of $\epsilon$.

We give an example in Fig.~\ref{fig:adaptive_grid}. Fig.~\ref{fig:adaptive_grid}~(a) shows the upper limit of $\epsilon$ in an auction with two additive bidders, two items, and  values with i.i.d.~uniform distribution on $[0,1]$.
 We fix the value of Bidder 2 to $v_2=(0, 0.6)$, and focus on the grid of Bidder 1's values. We do not consider values under which Bidder 1 prefers the IR option (0 allocation with 0 price). This is because, as discussed in Section~\ref{sec:method:rdc},
 these values are safe and cannot introduce infeasibility 
 as long as the price adjustments
 are non-negative.
Fig.~\ref{fig:adaptive_grid}~(b) shows that, if we use a uniform grid of $30\times 30$, then $\epsilon=1/30$ falls beneath the calculated upper limit for the majority of regions within this grid. This means in these regions the network is smooth with a low local Lipschitz constant, and we are free from incompatibility after transforming with a $30\times 30$ grid.

However, we only require that $\epsilon$ is \emph{just} below the surface shown here to guarantee menu compatibility over the entire value space, indicating that it is feasible to construct an adaptive grid  without compromising the menu compatibility guarantee.
%
As shown in Fig.~\ref{fig:adaptive_grid}~(c), the adoption of an adaptive grid results in the elimination of 95\% of the original
grid points, positioning $\epsilon$ just marginally below the upper limit surface throughout the grid.
Such an adaptive grid can greatly reduce the running time required by the price adjustment process as shown in Table.~\ref{tab:eff_strategy}. 

 \subsection{Neural Network Architecture and Training}
 \label{appx: nntraining}

 The neural Network architecture for both the allocation and payment networks uses two hidden layers, with each layer containing 1024 neurons. The architecture for the case with $n$ additive buyers and $m$ items is shown in Figure \ref{fig:gemnet_architecture} and the case for $n$ unit demand buyers and $m$ items is shown in Figure \ref{fig:ud_gemnet_architecture}. The activation function used in each  layer
 is the {\em Gaussian Error Linear Unit (GELU)}~\citep{hendrycks2016gaussian}. We applied spectral normalization to each layer of the network by normalizing the weights using the largest singular value of the weight matrix.
 \begin{figure}[t]
\centering
\scalebox{0.95}{
\includegraphics[width=\textwidth]{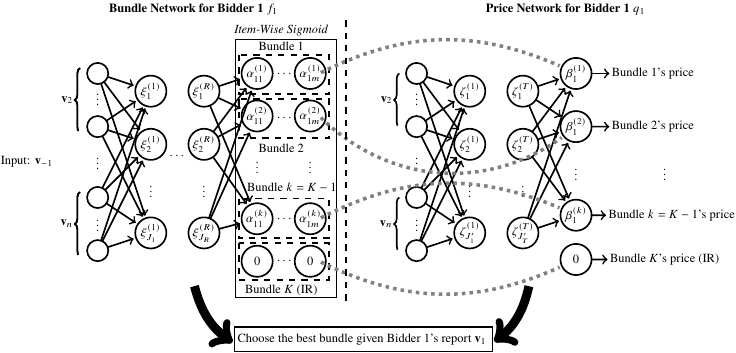}}
\caption{Neural network architecture of \name~ for Bidder 1 in the settings with $n$ additive
bidders and $m$ items. The inputs are the reported types $\mathbf{v}_{-1}$ of other bidders. The Bundle Network $f_1$ outputs $K$ menu items each of size $m$, and the Price Network $q_1$ outputs $K$ item prices each of size $1$. We add an \textit{item-wise sigmoid} on each element in the output of the Bundle Network, ensuring that the allocation probability is well-defined. The network architecture is the same for any other bidder $i'\neq 1$, where the inputs will be $v_{-i'}$, and the rest of the architecture remains the same. An agent $i$'s allocation and payment is the bundle that maximizes the utility for their report type $\mathbf{v}_i$. 
\label{fig:gemnet_architecture}}
\end{figure}

 \begin{figure}[t]
\centering
\scalebox{0.95}{
\includegraphics[width=\textwidth]{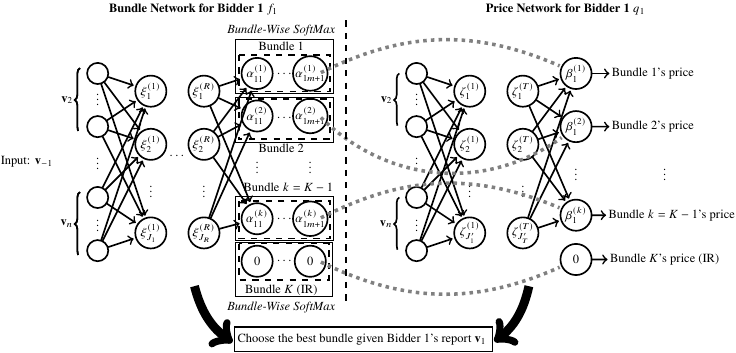}}
\caption{Neural network architecture of \name~ for Bidder 1 in the settings with $n$ unit demand
bidders and $m$ items. The inputs are the reported types $\mathbf{v}_{-1}$ of other bidders. The Bundle Network $f_1$ outputs $K$ menu items each of size $m+1$,
 and the Price Network $q_1$ outputs $K$ item prices each of size $1$. The output space include a dummy element $a_{1(m+1)}^{(k)}$ for each bundle $k$ (the output size of the network is thus $K\times (m+1)$), and we perform a \textit{row (bundle) -wise SoftMax} on the output of the Bundle Network, normalizing each bundle such that the sum of the selection probabilities for all items within a single bundle does not exceed 1. This can be viewed as a  lottery giving the bidder either one or no item (if the dummy item is chosen) according to the bundle probabilities, in line with the unit demand setup. 
 The network architecture is the same for any other bidder $i'\neq 1$, where the inputs will be $v_{-i'}$, and the rest of the architecture remains the same. An agent $i$'s allocation and payment is the bundle that maximizes the utility for their report type $\mathbf{v}_i$.
\label{fig:ud_gemnet_architecture}}
\end{figure}

We initialize the softmax lambda $\lambda_{\textsc{SoftMax}}$ in Eq.~\ref{equ:softmax_in_loss} as $5$, and the scaling factor $\lambda_{\textsc{Incomp}}$ in Eq.~\ref{equ:f_loss} as $0.1$. We train, using a minibatch of size $2^{13}$, for a fix number of iterations (e.g. $20,000$), and stop unless the violation rate is still above the desired threshold, in which case we continue to increase $\lambda_{\textsc{Incomp}}$ until violation rate goes below the desired level. During training, we gradually increase the SoftMax temperature $\lambda_{\textsc{SoftMax}}$ in Eq.~\ref{equ:softmax_in_loss}, as smaller values of $\lambda_{\textsc{SoftMax}}$ help with initial exploration of the network weights and larger values approximate the argmax operation better. We also gradually increase the scaling factor $\lambda_{\textsc{Incomp}}$. An example schedule we used is to increase $\lambda_{\textsc{SoftMax}}$ by 2 times per $3000$ until it reaches $2000$. We also increase $\lambda_{\textsc{Incomp}}$ by $\max(0.01, \mathcal{L}_{\textsc{Incomp}}(\theta)) .$

We evaluate the performance of the network once every $100$ epochs using a test set of $50K$ samples. Of all the checkpoints satisfying the an upper bar ($0.1\%$ or $0.5\%$,  depending on task) of feasibility violation, we pick the one with the best revenue for the subsequent menu transformation.

We normally set the number of menus to be $K=300$, as empirically this yields good performance. When the number of items is large $>5$, we need to increase this number (e.g. $K=1000$) to capture a more rich space of allocation over items. We use Adam as the optimizer and a learning rate of $0.005$. When the revenue and violation rate starts to converge (change in revenue is $\leq 0.03$ and change in violation is $\leq 0.01$) across $2$ consecutive validation evaluations (which occurs once every $200$ epochs), we initiate a linear decay of the learning rate, decreasing it by a factor of $1/10$ every 2,000 iterations.

On a typical task of 3 bidders and 2 items, the training time for $1000$ iterations is an average of $50.88$ seconds. The results provided in the training are typically obtained on $\sim 20,000$ iterations, which takes about $15$ minutes. Compared to RegretNet \citep{duetting2023optimal}, we do not need to run an inner loop of gradient descent to find optimal misreport in order to calculate the regret, resulting in a faster runtime.
We note, though, that the exact training time varies according our objective. A result close to that of baselines (AMenuNet \citep{duan2023scalable} specifically) can typically be obtained in $\sim 5000$ iterations, taking less than 5 minutes. To achieve notable improvement over the baselines, on the other hand, might take longer training time (up to a few hours). We can control this tradeoff at our discretion. \nocite{song2019convergence}

\end{document}